\documentclass[letterpaper, 10 pt, final]{IEEEtran}

\IEEEoverridecommandlockouts

\usepackage{amssymb,amsfonts,amsthm}
\usepackage{eucal,bm}
\usepackage{gensymb,subfigure}
\usepackage{booktabs}
\usepackage{times,color,url}
\usepackage{comment}

\usepackage{cite}

\usepackage{array}

\usepackage{graphicx}
\usepackage{epstopdf}

\usepackage[cmex10]{amsmath}
\usepackage{enumitem}

\newcommand{\f}[2]{\frac{#1}{#2}}

\newcommand{\de}[2]{\frac{d #1}{d #2}} 
\newcommand{\dd}[2]{\frac{d^2 #1}{d #2^2}} 
\newcommand{\h}[0]{{\mathcal R \mathcal H}_\infty}
\renewcommand{\d}[0]{\mathbb{D}}
\renewcommand{\c}[0]{\partial \mathbb{D}}
\newcommand{\dc}[0]{\overline{\mathbb{D}}}
\newcommand{\p}[0]{\mathcal{P}}

\DeclareMathOperator*{\argmax}{arg\,max}

\DeclareMathOperator*{\mcup}{\cup}

\newtheorem{theorem}{Theorem}
\newtheorem{lemma}{Lemma}
\newtheorem{corollary}{Corollary}

\allowdisplaybreaks

\title{Approximation by Simple Poles -- Part I: Density and Geometric
  Convergence Rate in Hardy Space}

\author{Michael W. Fisher, Gabriela Hug, and Florian D\"{o}rfler
  \thanks{This paper is based upon work supported by the King
  Abdullah University of Science and Technology (KAUST) Office of
  Sponsored Research (award No. OSR-2019-CoE-NEOM-4178.11) and by the
  European Union’s Horizon 2020 research and innovation program (grant
  agreement No. 883985).

  F. D\"{o}rfler and G. Hug are with
  ETH Z\"{u}rich, 8092 Z\"{u}rich, Switzerland.
  
  M. W. Fisher is with University of Waterloo, Waterloo, Ontario, Canada.

  Email: \{mfisher, dorfler\}@ethz.ch; hug@eeh.ee.ethz.ch.
}
}

\begin{document}

\maketitle








\begin{abstract}
  Optimal linear feedback control design is a valuable but challenging problem
  due to nonconvexity of the underlying optimization and infinite
  dimensionality of the Hardy space of stabilizing controllers.
  A powerful class of techniques for solving optimal control problems involves
  using reparameterization to transform the control design to a convex
  but infinite dimensional optimization.
  To make the problem tractable, historical work focuses on Galerkin-type
  finite dimensional approximations to Hardy space, especially those involving
  Lorentz series approximations such as the finite impulse response
  appproximation.
  However, Lorentz series approximations can lead to infeasibility,
  difficulty incorporating prior knowledge, deadbeat control in the case of
  finite impulse response, and increased suboptimality, especially for systems
  with large separation of time scales.
  The goal of this two-part article is to introduce a new Galerkin-type
  method based on approximation by transfer functions with a selection of
  simple poles,
  and to apply this simple pole approximation for optimal control design.
  
  In Part I, error bounds for approximating arbitrary transfer functions in
  Hardy space are provided based on the geometry of the pole selection.
  It is shown that the space of transfer functions with these simple poles
  converges to the full Hardy space, and a uniform convergence rate is provided
  based purely on the geometry of the pole selection.
  This is then specialized to derive a convergence rate for a particularly
  interesting pole selection based on an Archimedes spiral.
  In Part II, the simple pole approximation is combined with system level
  synthesis, a recent reparameterization approach, to develop a new control
  design method.
  This technique is convex and tractable, always feasible, can include prior
  knowledge, does not result in deadbeat control, and works well for systems
  with large separation of time scales.
  A suboptimality certificate is provided which bounds the suboptimality based
  on the geometry of the pole selection.
  This is then specialized to obtain a convergence rate for the Archimedes
  spiral pole selection.
  An example demonstrates superior performance of the method.
\end{abstract}




      

\section{Introduction}\label{sec:intro}


Optimal control design for linear time invariant
systems has been thoroughly studied, but remains a difficult problem.
Key challenges arise from the infinite dimensionality of the
Hardy space of stabilizing controllers, as well as the nonconvexity of the
resulting optimization problem for the design.

One of the most celebrated approaches for optimal control design is the Youla
parameterization \cite{Yo76},
which parameterizes all stabilizing controllers and leads
to a convex reformulation in terms of the Youla parameter.
Recently, new optimal control design techniques have been introduced, known
as system level synthesis (SLS) \cite{Wa18,Wa19} and input-output
parameterization (IOP) \cite{Fu19}, 
which also parameterize all stabilizing controllers
and include closed-loop system responses as decision variables.
However, they have some advantages over Youla parameterization, including
a larger class of optimal control problems which admit convex
representations,
and the ability to preserve structure from the closed-loop transfer functions in
the internal controller realizations.



Applying these parameterizations to mixed
$\mathcal{H}_2/\mathcal{H}_\infty$ control design results in a convex but
infinite dimensional optimization problem since the decision variables are
transfer functions lying in an infinite dimensional Hardy space.
In order to solve these problems, historical work has
focused on using finite dimensional approximations of Hardy space based on
Lorentz series approximations \cite[Chapter~15]{Bo91},\cite{Fu19,Li99} to
obtain a tractable optimization for the design.
In particular, in discrete time the Lorentz series approximation with all
poles at the origin, known as the finite impulse response (FIR) approximation,
is the most popular choice, and is commonly used with Youla parameterization
\cite{Qi04,Al13}, IOP \cite{Fu19}, and SLS \cite{An19}.

The rate at which the Lorentz series approximation converges determines
the number of poles required to achieve a close approximation
and, hence, acceptable suboptimality.
For Lorentz approximations, this convergence rate typically
depends on the particular optimal transfer function, and therefore may result
in large numbers of poles to achieve acceptable performance.
This is an important practical concern, as large numbers of poles can lead to
high computational complexity for the control design, lack of robustness in
the resulting controller, and implementation challenges in practice
\cite[Chapter 19]{Zh95}.
For FIR, the number of poles is equal to the length of the FIR, so this
becomes especially problematic when the optimal transfer function has a long
settling time, such as in systems with large separation of time
scales, where short sampling times are needed to capture the fast dynamics,
which are also coupled with much slower dynamics.

Lorentz series approximations can also lead to
infeasibility of the control design because all the poles in the approximation
are at the same location.
This occurs for example when the optimal
transfer function is proportional to the closed-loop transfer function,
and the plant has stable but uncontrollable poles.
In some cases it may be possible to recover feasibility, such as in
SLS by introducing a slack variable which enables constraint violation
\cite[Section~4.5.3]{An19}, but this leads to additional suboptimality
\cite[Theorem~4.7]{An19} and results in a quasi-convex optimization which
requires an iterative approach such as golden section search to solve.

Furthermore, since all the poles in the Lorentz series approximation are at
the same location, it is unclear how to incorporate prior knowledge about
the optimal solution into the control design.
For example, prior information about the locations of some of the optimal
closed-loop poles
may be available, such as for model matching 
\cite{Sh09},
model reference control \cite{Ab09}, design based on the internal model
principle \cite{Fr76}, expensive control \cite[Theorem~3.12(b)]{Kw72}, etc.
In such situations, it would be desirable to use this information to improve
performance, but it is not clear how to do so with the Lorentz approximation.

In addition, FIR in particular results in deadbeat control (DBC), which often
experiences poorly
damped oscillations between discrete sampling times that can even
persist in steady state, as well as lack of robustness to model uncertainty and
parameter variations because of the high
control gains required to reach the origin in finite time \cite{Ur87}.
So, it is often advantageous to avoid deadbeat control when possible.

This two-part work develops a new Galerkin-type method based on approximation
by transfer functions with simple poles, i.e., poles with multiplicity no more
than one, and uses it to develop a novel optimal control
design method to address the limitations of the Lorentz approximation approach.
The main contributions of our two-part paper are as follows.

In Part I, the simple pole approximation (SPA) is introduced, which gives
the designer the freedom to form an approximation using any finite
selection of stable poles in the unit disk that are closed under complex
conjugation.
It is shown that, for any transfer function in Hardy space, the
approximation error from SPA is bounded, and this bound is proportional
to the geometric distance between the
poles of SPA and the poles of the desired transfer function.
Combining these approximation error bounds with the notion of a space-filling
sequence of poles in the unit disk, it is shown that the subspace of SPAs
converges to the entire Hardy space.
In particular, it is possible to approximate any transfer function in Hardy
space arbitarily well with respect to both the $\mathcal{H}_2$ and
$\mathcal{H}_\infty$ norms using SPA.
Furthermore, if the poles of a space-filling sequence converge to the entire
unit disk at a certain rate, which we term a geometric convergence rate,
then it is shown that SPA converges to any transfer function in Hardy space
at this same uniform convergence rate, which depends purely on the geometry
of the pole selection.
Therefore, unlike with Lorentz approximations in general, and FIR in
particular, the convergence rate of SPA is not reduced for transfer functions
with long settling times, such as those resulting from systems with large
separation of time scales.
Finally, a particularly interesting space-filling sequence of poles is
introduced based on the Archimedes spiral, and used to provide an explicit
uniform convergence rate for this special case based on the spiral geometry.

In Part II \cite{Fi22b}, SPA is combined with SLS to develop a new control
design method.
The approach has a uniform convergence rate, with number of poles
independent of the settling time of the optimal closed-loop responses, and
so works well for systems with large separation of time scales.
It automatically ensures feasibility for stabilizable plants by incorporating
all of the plant poles into the pole selection.  Therefore, it does not
require additional slack variables or constraint violations,
and additional suboptimality resulting from these can be avoided.
Furthermore, it results in a convex optimization for the control design which
can be solved with a single semidefinite program, so no iterative algorithms
are required.  If prior information about the locations of some of the optimal
closed-loop poles is known in advance, these can be directly incorporated
into the pole selection of SPA, leading to improved performance.
In addition, SPA is not FIR, so it does not result in deadbeat control.
A suboptimality certificate is provided which gives a uniform bound on the
suboptimality based on the geometry of the SPA pole selection.
This is then specialized to the pole selection of the Archimedes spiral to
obtain an explicit convergence rate.
Superior performance of the method is demonstrated
on an example of power converter control design, which first motivated the
development of SPA due to inadequate performance of SLS with FIR.
This example is fully reproducible with all code publicly
available \cite{git_sls}.

The main connection between Part I and Part II is that the control design method
and suboptimality certificates developed in Part II rely heavily on the SPA
method, bounded approximation error results, and convergence certificates
developed in Part I.
Furthermore, the example of Part II demonstrates the practical use of the
SPA method from Part I for control design, and its superior performance
compared to the FIR approach for that case.
As the SPA method and results provided in Part I are developed in generality,
and not specialized to SLS with state feedback until Part II, they may be of
independent interest for other control design techniques as well, such as IOP,
Youla parameterization, and SLS with output feedback.
In turn, the derivation of the suboptimality bounds in Part II may provide a
general recipe for deriving similar suboptimality bounds for these other control
design methods as well.

The remainder of Part I is organized as follows.
Section~\ref{sec:not} provides some notation and preliminaries.
Section~\ref{sec:res} provides the development
of the SPA method as well as the main results regarding bounded approximation
error, convergence of the SPA to all of Hardy space, uniform convergence rate,
and specialization of these to the Archimedes spiral pole selection.
Section~\ref{sec:proof} gives the proofs of the theoretical results.
Finally, Section~\ref{sec:con} offers concluding remarks.

\section{Preliminaries and Notation}\label{sec:not}



We will use the following notation throughout the paper.
Let $\mathbb{D}$ be the open unit disk in the complex plane,
$\overline{\mathbb{D}}$ the
closed unit disk, and $\partial \mathbb{D}$ the unit circle.
Let $\overline{B}_r$ be the closed ball of radius $r$ centered at the origin.
For any set $S \subset \mathbb{C}$ and any $z \in \mathbb{C}$,
let $d(z,S) = \inf_{w \in S} |z-w|$.
For any sets $S, S' \subset \mathbb{C}$,
let $d(S,S') = \inf_{z \in S, w \in S'} |z-w|$.
For any set $S \subset \mathbb{C}$ and $\epsilon > 0$, define the $\epsilon$
neighborhood of $S$, denoted $S_\epsilon$, to be the set
$\{z \in \mathbb{C}:~d(z,S) \leq \epsilon\}$.
For any nonempty sets $S, S' \subset \mathbb{C}$, define the Hausdorff distance
$d_H(S,S')$ to be the infimum over $\epsilon > 0$ such that
$S \subset S'_\epsilon$ and $S' \subset S_\epsilon$.
Then $d_H$ is a well-defined metric, and we say that a sequence of sets
$\{S_n\}_{n=1}^\infty$ converges to the set $S^*$, denoted $S_n \to S^*$,
if $\lim_{n \to \infty} d_H(S_n,S^*) = 0$.
For any $z \in \mathbb{C}$, let $\overline{z}$ denote the complex conjugate
of $z$, and let $\text{Re}(z)$ and $\text{Im}(z)$ denote the real and
imaginary components of $z$, respectively.
For any positive integer $m$, let $I_m = \{1,2,...,m\}$.
For any set $S$, let $|S|$ be the cardinality of set $S$ (i.e., the number of
elements it contains).

For a matrix $M$, let $||M||_2$ denote the spectral norm of $M$,
and let $||M||_F$ denote the Frobenius norm of $M$.
Let $M_{i,j}$ denote the entry of $M$ at row $i$ and column $j$.


Let $\h$ be the Hardy space of real, rational, proper, and stable transfer
functions in discrete time (see \cite{Zh95,Du00} for a more detailed
discussion).  Let $\hat{m} \times \hat{n}$ denote the dimensions of each
transfer function in $\h$.
Let $\f{1}{z} \h \subset \h$ consist of the strictly proper transfer functions
in $\h$.
For $S \in \h$, define the norms
\begin{align*}
  ||S||_{\mathcal{H}_\infty} &= \sup_{z \in \c} ||S(z)||_2, \;
  ||S||_{\mathcal{H}_2}^2 = \f{1}{2\pi} \int_{z \in \c} ||S(z)||_F^2 ~dz.
\end{align*}
For any $S, S' \in \h$, define the metrics
\begin{align*}
  d_{\mathcal{H}_\infty}(S,S') &= ||S-S'||_{\mathcal{H}_\infty}, \;
  d_{\mathcal{H}_2}(S,S') &= ||S-S'||_{\mathcal{H}_2}.
\end{align*}
For any $S \in \h$ and positive integer $k$, let $\mathcal{I}(S)(k)$ denote
the time-domain impulse response of $S$ at time $k$.
For $S \in \h$, by Parseval's theorem it can be shown that
\begin{align*}
  ||S||_{\mathcal{H}_2}^2 &= \sum_{k=1}^\infty ||\mathcal{I}(S)(k)||_F^2
    =: ||\mathcal{I}(S)||_F^2 = \lim_{T \to \infty} ||\mathcal{I}_T(S)||_F^2,
\end{align*} 
where ${\scriptstyle \mathcal{I}_T(S)} := \left[\begin{smallmatrix}
    S(1)^\intercal &
    S(2)^\intercal &
    \hdots & 
    S(T)^\intercal
  \end{smallmatrix}\right]^\intercal$,
which can be well approximated numerically using $T$ finite.
For any $S \in \h$ let $\mathcal{C}(S)$ denote the convolution operator of
$S$ so that for any input signal $u(k)$ and nonnegative integer $n$,
$\mathcal{C}(S)(u)[n] = \sum_{k=1}^\infty \mathcal{I}(S)(k)u(n-k).$
Note that for $S \in \h$, since $||\cdot||_{\mathcal{H}_\infty}$ is the
induced $\mathcal{L}_2 \to \mathcal{L}_2$ norm, it can be shown that
\begin{align*}
  ||S||_{\mathcal{H}_\infty} &= \mkern-10mu
  \sup_{||u||_2 \neq 0} ||\mathcal{C}(S)(u)||_2
    =: ||\mathcal{C}(S)||_2 = \lim_{T \to \infty} ||\mathcal{C}_T(S)||_2 \\
  {\scriptstyle \mathcal{C}_T(S)} &:=
   {\scriptstyle
  \left[\begin{smallmatrix}
    S(1) &  0  & 0 & 0 \\
    S(2) & S(1) & 0 & 0 \\
    \vdots &  & \ddots & 0 \\
    S(T) & S(T-1) & ... & S(1) \\
    \end{smallmatrix}\right],}
\end{align*}
which can be well approximated numerically using $T$ finite.

We consider another notion of set convergence, which will be useful when
proving density of the simple pole approximation in Hardy space.
Let $U$ be a metric space with metric $\hat{\rho}$, let $S \subset U$, and
let $\{S_n\}_{n=1}^\infty$ be a sequence of subsets of $U$.
Define $\liminf_{n \to \infty} S_n$ with respect to $\hat{\rho}$ to be the set of
points $x \in U$ such that
there exists a sequence $\{x_n\}_{n=1}^\infty$ with $x_n \in S_n$ for all $n$
such that $\lim_{n \to \infty} x_n = x$ with respect to $\hat{\rho}$
(i.e., $\lim_{n \to \infty} \hat{\rho}(x_n,x) = 0$).
Define $\limsup_{n \to \infty} S_n$ with respect to $\hat{\rho}$ to be the set of
points $x \in U$ such that
there exists a subsequence $\{S_{n_m}\}_{m=1}^\infty$ of $\{S_n\}_{n=1}^\infty$
and a sequence $\{x_{n_m}\}_{m=1}^\infty$ such that $x_{n_m} \in S_{n_m}$ for all
$m$ and $\lim_{m \to \infty} x_{n_m} = x$ with respect to $\hat{\rho}$.
If $\liminf_{n \to \infty} S_n = \limsup_{n \to \infty} S_n$, say
$\liminf_{n \to \infty} S_n = \limsup_{n \to \infty} S_n = S$,
then we say that
$\{S_n\}_{n=1}^\infty$ converges to $S$, and write $\lim_{n \to \infty} S_n = S$,
with respect to $\hat{\rho}$.
Note that in the special case where $U$ is compact,
$\lim_{n \to \infty} S_n = S$ if and only if $\lim_{n \to \infty} d_H(S_n,S) = 0$
\cite[Section~28, Statement~V]{Ha57},
so this notion of set convergence is a natural extension to noncompact metric
spaces of convergence with respect to the Hausdorff distance.

\section{Main Results}\label{sec:res}

\subsection{Approximation by Simple Poles}
\label{sec:approx}

We begin by showing that transfer functions with only simple poles (i.e., poles
with multiplicity no greater than one) can be used
to approximate any transfer function in $\f{1}{z}\h$, including those with
repeated poles, to arbitrary accuracy.
Towards that end, let $S \in \f{1}{z}\h$ and let $\mathcal{Q}$ be the poles
of $S$.
For each pole $q \in \mathcal{Q}$, let $m_q^*$ be its multiplicity in $S$
and let $m_{\max}$ be the maximum multiplicity,
i.e., $m_{\max} = \max_{q \in \mathcal{Q}} m_q$.
Then the partial fraction decomposition of $S$ can be written
\begin{align*}
  S(z) = \sum_{q \in \mathcal{Q}} \sum_{j=1}^{m_q^*} G_{(q,j)}^* \f{1}{(z-q)^j}
\end{align*}
for some constant coefficient matrices $G_{(q,j)}^*$.

Let $\p$ be a set of simple and distinct poles, hereafter referred to as
approximating poles, which will be used to construct a transfer function for
approximating $S$ (i.e., $S \approx \sum_{p \in \p} G_p \f{1}{z-p}$).
The key idea is that for each pole $q \in \mathcal{Q}$,
an approximating transfer function is constructed to approximate $q$'s
contribution to the partial fraction decomposition of $S$. The
poles of this approximation are selected to be the $m_q$ closest poles in
$\p$ to $q$, which we
denote by $\p(q)$.
Then, the overall approximating transfer function for $S$ is obtained by
summing over the individual approximating transfer functions for each
$q \in \mathcal{Q}$.

To evaluate the accuracy of the approximation, for each $q \in \mathcal{Q}$
let $\hat{d}(q)$ be the distance from $q$ to the furthest of the $m_q$ simple
poles being used to approximate it, i.e.,
$\hat{d}(q) = \max_{p \in \p(q)} |p-q|$.
Let $D(\p)$ be the maximum of these distances over all the poles in
$\mathcal{Q}$, i.e., $D(\p) = \max_{q \in \mathcal{Q}} \hat{d}(q)$.
Then $D(\p)$ represents the largest distance between approximating poles
in $\p$ and the poles in $\mathcal{Q}$ they are being used to approximate,
so it measures the worst-case geometric error in this pole approximation.
Intuitively one might therefore expect that as $D(\p) \to 0$, the approximating
transfer function would approach $S$.
This intuition is formalized in
Theorem~\ref{thm:approx}, which provides an
approximation error bound in terms of standard Hardy space norms
of the
simple pole approximation which is linear in $D(\p)$.
Thus, Theorem~\ref{thm:approx} shows that the simple pole approximating
transfer function converges to $S$ at least linearly with $D(\p)$ and,
therefore, that this convergence rate depends purely
on the geometry of the pole selection. 

Before presenting Theorem~\ref{thm:approx}, we make the following assumptions
regarding the set of approximating poles $\p$:

  \vspace{5pt}

{\it
  
(A1) There exists some $r \in (0,1)$ such that
$\p \subset \overline{B}_r$.

(A2) $|\p| \geq m_{\max}$. 

(A3) $\p$ is closed under complex conjugation (i.e., $p \in \p$
\indent implies that $\overline{p} \in \p$).

(A4) $D(\p) < 1$.

(A5) Let $\sigma \subset \d$ be finite.  Then for every
$q \in \mathcal{Q}$ and every \indent $\lambda \in \sigma$ with
$\lambda \neq q$, $\lambda \not\in \p(q)$.
}

  \vspace{5pt}

Assumption A1 ensures that $\p$ consists of stable poles,
Assumption A2 that the size of $\p$ is at least as large as $m_{\max}$,
Assumption A3 that $\p$ can be used to construct a transfer function with
real coefficients (as will be seen in the proof of Theorem~\ref{thm:approx}),
and Assumption A4 that $\p$ is not excessively far from $\mathcal{Q}$.
Note that Assumption A4 can be satisfied with only two pairs
of complex conjugate poles if $S$ has only simple poles,
so it tends not to be restrictive in practice.
Assumption A5 is introduced here in anticipation of its use in
\cite[Theorem~1]{Fi22b} where $\sigma$ will represent the poles of the plant,
and can be easily satisfied in practice.
For its use in \cite[Theorem~1]{Fi22b}, define
$\delta = \min_{q \in \mathcal{Q}, \lambda \in \sigma, \lambda \neq q} d(\lambda,\p(q))$,
and note that $\delta > 0$ by Assumption A5.
We are now ready to present Theorem~\ref{thm:approx}.

\begin{theorem}[Simple Pole Approximation]\label{thm:approx}
Let $S \in \f{1}{z}\h$ and let $\p$ be a set of poles satisfying Assumptions
A1-A4.  
Then there exist constants
$c_S = c_S(\mathcal{Q},G_{(q,j)}^*,r) > 0$ and
$c_S' = c_S'(\mathcal{Q},G_{(q,j)}^*,r) > 0$,
and constant matrices
$\{G_p\}_{p \in \p}$ such that
$\sum_{p \in \p} G_p \f{1}{z-p} \in \f{1}{z}\h$
and 
\begin{align}
  \begin{split}
    \left|\left| \sum_{p \in \p} G_p \f{1}{z-p} - S
    \right|\right|_{\mathcal{H}_2}
    &\leq c_S D(\p) \\
    \left|\left| \sum_{p \in \p} G_p \f{1}{z-p} - S
    \right|\right|_{\mathcal{H}_\infty}
    &\leq c_S' D(\p).
  \end{split}
  \label{eq:mimo_approx}
\end{align}
\end{theorem}

Note that the constants $c_S$ and $c_S'$ appearing in
Theorem~\ref{thm:approx} depend only on
$S \in \f{1}{z}\h$ and on the radius
$r$ of the closed ball $\overline{B}_r$ in which $\p$ is contained, and
do not otherwise depend on the specific pole selection $\p$.
This feature will play a crucial role in the proofs of the density and
uniform convergence results of Section~\ref{sec:dens}.

\subsection{Density and Uniform Converence Rate}\label{sec:dens}

The goals of this section are to show that simple pole approximations using
suitable pole selections are dense in the Hardy space $\f{1}{z}\h$, and that
a uniform convergence rate of the simple pole approximation to any transfer
function in $\f{1}{z}\h$ can be provided which depends only on the geometry of
the pole selection.
Towards that end, we define a sequence of poles, denoted $\{\p_n\}_{n=1}^\infty$,
to be a sequence where for each $n$, $\p_n$ is a finite collection of poles
contained in $\d$ and closed under complex conjugation.
We say that a sequence of poles $\{\p_n\}_{n=1}^\infty$ exhibits
{\it geometric convergence} if for any $S \in \f{1}{z}\h$,
its worst-case geometric approximation error converges to zero,
i.e., $\lim_{n \to \infty} D(\p_n) = 0$.
This definition relates naturally to the approximation error bound provided
in Theorem~\ref{thm:approx}.
However, it is valuable to introduce the related notion of a space-filling
sequence of poles.
In particular, we say that a sequence of poles $\{\p_n\}_{n=1}^\infty$ is
{\it space-filling}
if it is a sequence of pole
selections such that $\p_n \to \dc$ with respect to the Hausdorff distance,
i.e., the pole sequence converges to the entire unit disk.
A particularly interesting example of a space-filling sequence based on
the Archimedes spiral is provided in Section~\ref{sec:spiral}.

Unlike geometric convergence, the notion of space-filling is more
intuitive to visualize and understand, avoids an abstract definition in
terms of transfer functions, is easier to establish in practice, and
will serve to illuminate the connection between convergence of the pole
sequence to the entire unit disk and convergence of the space of simple pole
approximations to the entire Hardy space $\f{1}{z}\h$.
Therefore, it is natural that Theorem~\ref{thm:dens}, which establishes the
convergence of the space of simple pole
approximations to $\f{1}{z}\h$, is shown for any space-filling sequence
of poles.

Lemma~\ref{lem:equiv} shows that for a sequence of poles, the topological
property of space-filling is equivalent to the geometric property of
exhibiting geometric convergence.
This will be useful in establishing density of the simple pole approximation
using space-filling sequences of poles in the Hardy space $\f{1}{z}\h$.

\begin{lemma}[Equivalence of Space-Filling and Geometric Convergence]
  \label{lem:equiv}
  Let $\{\p_n\}_{n=1}^\infty$ be a sequence of poles.
  Then $\{\p_n\}_{n=1}^\infty$ is space-filling if and only if it exhibits
  geometric convergence.
\end{lemma}

For any sequence of poles $\{\p_n\}_{n=1}^\infty$, for each $n$ let
\begin{align*}
  \mathcal{A}_n = \left\{\sum_{p \in \p_n} G_p \f{1}{z-p} \in \f{1}{z}\h :~
  G_p \in \mathbb{C}^{\hat{n} \times \hat{m}}\right\}.
\end{align*}
Then $\mathcal{A}_n \subset \f{1}{z}\h$ denotes the space of
simple pole approximations resulting from the pole selection $\p_n$.
Theorem~\ref{thm:dens} shows the density
of the simple pole approximation with any space-filling sequence of poles
in the Hardy space $\f{1}{z}\h$.
For the statement of Theorem~\ref{thm:dens}, it is useful to recall the notion
of set convergence defined in Section~\ref{sec:not}.

\begin{theorem}[Density in Hardy Space]\label{thm:dens}
  Let $\{\p_n\}_{n=1}^\infty$ be a space-filling sequence of poles.
  Then
  \begin{align*}
    \lim_{n \to \infty} \mathcal{A}_n  =  \f{1}{z}\h
  \end{align*}
  with respect to $d_{\mathcal{H}_2}$ and $d_{\mathcal{H}_\infty}$.
  
\end{theorem}

For any $k > 0$, we say that a sequence of poles $\{\p_n\}_{n=1}^\infty$ has
geometric convergence
rate $\f{1}{n^k}$ if for each $S \in \f{1}{z}\h$ there exists a constant
$c_S > 0$ such that $D(\p_n) \leq \f{c_S}{n^k}$ for all positive integers $n$.
Note that this is in fact a uniform convergence rate since the rate $k$ is
independent of the choice of $S \in \f{1}{z}\h$.
Theorem~\ref{thm:conv} shows that if a sequence of poles has geometric
convergence rate $\f{1}{n^k}$, then for any $S \in \f{1}{z}\h$ the simple
pole approximation converges to $S$ in the Hardy space norms at the rate
$\f{1}{n^k}$.

\begin{theorem}[Uniform Convergence Rate]\label{thm:conv}
  For some $k > 0$, let $\{\p_n\}_{n=1}^\infty$ be a sequence of poles with
  geometric convergence rate $\f{1}{n^k}$.
  Then for any $S \in \f{1}{z}\h$, there exist constants
  $c_S = c_S(\mathcal{Q},G_{(q,j)}^*) > 0$,
  $c_S' = c_S'(\mathcal{Q},G_{(q,j)}^*) > 0$, and $N > 0$ such that for any
  $n \geq N$ there exist $\{G_p^n\}_{p \in \p_n}$ such that
  $\sum_{p \in \p_n} G_p^n \f{1}{z-p} \in \f{1}{z}\h$ and
  \begin{align*}
    \left|\left|\sum_{p \in \p_n} G_p^n \f{1}{z-p} -S\right|\right|_{\mathcal{H}_2}
    &\leq \f{c_S}{n^k} \\
   \left|\left|\sum_{p \in \p_n} G_p^n \f{1}{z-p} -S\right|\right|_{\mathcal{H}_\infty}
    &\leq \f{c_S'}{n^k}.
  \end{align*}
\end{theorem}

Note that $N$ in Theorem~\ref{thm:conv} only needs to be chosen to ensure
$D(\p_N) < 1$ and $|\p_N| \geq m_{\max}$,
which if $S$ has only simple poles can be satisfied with only two pairs of
complex conjugate poles (i.e., $N = 4$),
so it tends not to be restrictive in practice.

\subsection{Archimedes Spiral Pole Selection}\label{sec:spiral}

As discussed in Section~\ref{sec:intro}, many control design approaches
to 
mixed $\mathcal{H}_2/\mathcal{H}_\infty$ control design (see \cite[Eq.~2]{Fi22b}
for a detailed formulation)
require Galerkin-type finite dimensional approximations of
Hardy space $\f{1}{z}\h$.
Such methods implicitly seek a ground-truth optimal transfer function lying
in $\f{1}{z}\h$, which is in general unknown a priori.
By Theorem~\ref{thm:dens}, using the simple pole approximation (SPA) with any
space-filling sequence of poles is guaranteed to converge to this optimal
transfer function as the number of poles in the approximation increases.

In this section we provide recommendations for achieving fast
convergence rates, which typically leads to improved performance of the
control design.
Motivated by Section~\ref{sec:dens}, to use SPA
as well as possible the goal is to choose a pole selection
whose poles are as close as possible to the poles of an {\em optimal} transfer
function.
In the cases of system level synthesis (which will be explored further in
Part II \cite{Fi22b}) and input-output parameterization, these optimal
transfer functions are optimal closed-loop responses. 
Let $\mathcal{Q}$ denote the poles of the optimal closed-loop transfer function
from disturbance to control input (see $T_{v \to u}(z)$ in
\cite[Section II.A]{Fi22b} for a detailed definition), and let $\sigma$ denote
the poles of the plant.
When the control synthesis is feasible,
optimal closed-loop transfer functions exist
(see \cite[Assumption~A6]{Fi22b} and the following discussion) and have
a finite collection of poles, all of which are contained in
$\mathcal{Q} \cup \sigma$ (by the proof of \cite[Lemma~1]{Fi22b}).
So, in many cases it is desirable to choose a pole selection as close as
possible to $\mathcal{Q} \cup \sigma$.
Therefore, we recommend to first include the poles of the plant $\sigma$
in SPA, as well as any optimal poles which are
known a priori, such as discussed in Section~\ref{sec:intro}.

Once the prior information about the optimal and plant poles have been
incorporated into SPA, if no information about the remaining optimal
poles is known in advance, a
natural choice for these remaining poles in SPA would be to
distribute the poles evenly over
the unit disk so as to minimize $D(\mathcal{P}_n)$
(i.e., finding $\mathcal{P}_n$ for each $n$ such that $D(\mathcal{P}_n)$ is
minimized for the given number of poles in $\mathcal{P}_n$).
However, finding an exactly even pole distribution over the unit disk
is equivalent to finding the minimum energy configuration of a collection
of identical point charges over a disk.  This is a nonconvex optimization
problem that 
has never been solved in the
general case, requires high computational effort to solve even for smaller
numbers of poles, and for which even when solutions can be obtained
they are not guaranteed to be globally optimal \cite{Nu98}.
Therefore, rather than attempting to find exactly even pole distributions,
we resort to finding approximate solutions instead.
These approximations also have the advantage that, unlike for the exact
solution methods, they can be used to derive approximation error bounds
with a geometric convergence rate based on Theorem~\ref{thm:conv}.

\begin{figure}
  \centering
  \includegraphics[width=0.3\textwidth]{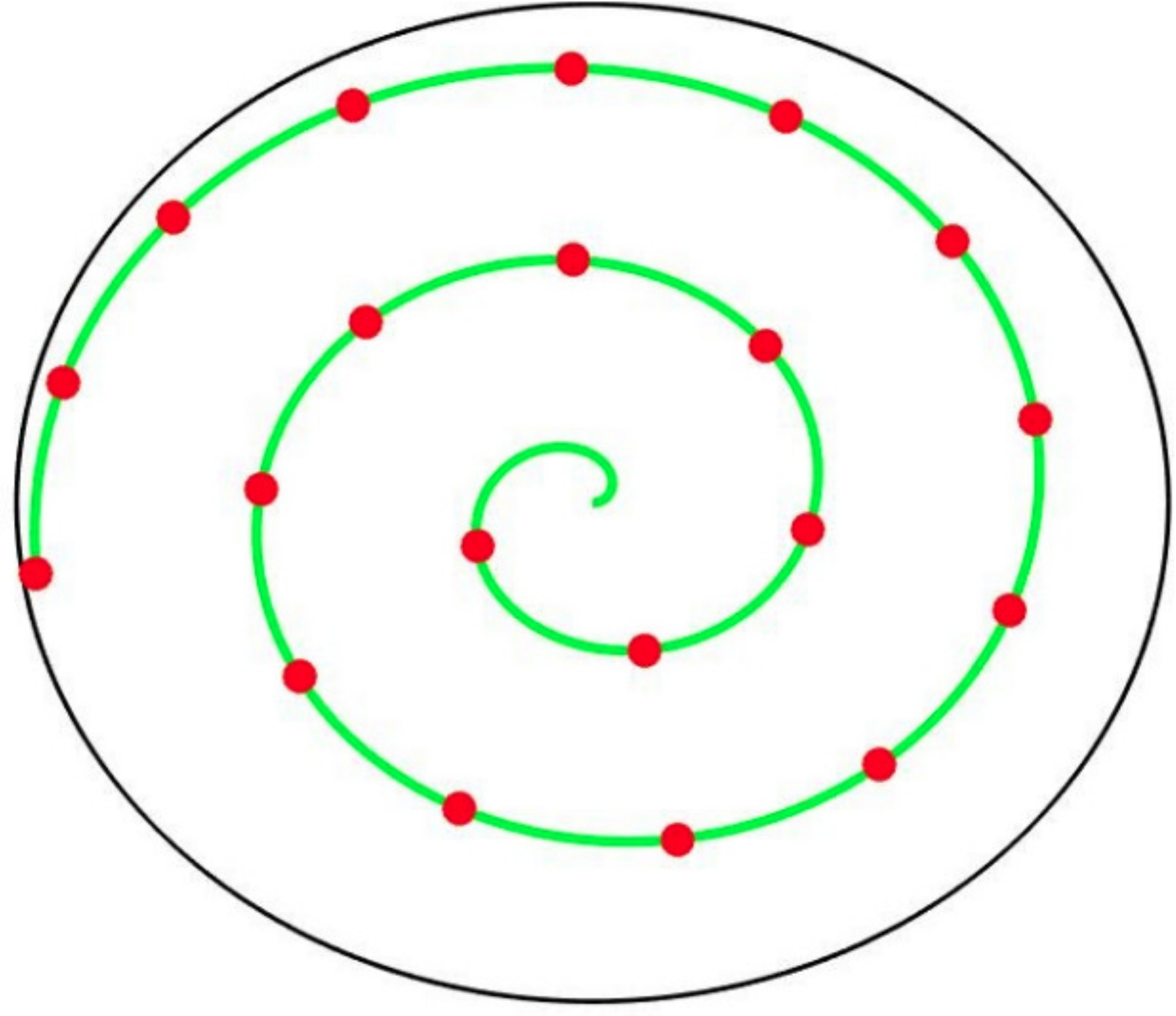}
  \caption{Archimedes spiral (green) with pole selection (red) to distribute
    the poles approximately evenly over the unit disk.  The selection
    in \eqref{eq:poles} also includes the complex conjugates of the pole
    selection shown.}
  \label{fig:spiral}
\end{figure}

One class of heuristic techniques that has been used to generate
appproximately evenly
spaced points over a disk, and that has worked well in practice in a variety
of different fields, is to select points along a spiral \cite{Ga10,Vo79,Ar15},
transforming the pole selection problem from a two-dimensional selection over
a disk to a one-dimensional selection along a spiral curve.
After selecting a particular spiral, typically specified in polar coordinates
by $r = r(\theta)$, appropriate points can then be chosen along the spiral to
yield an approximately even distribution.
Therefore, we choose the poles in SPA along an {\it Archimedes spiral}
according to
the selection proposed in \cite[Section~5.1]{Ar15} and as shown in
Fig.~\ref{fig:spiral}.
This yields 
a powerful heuristic
for minimizing $D(\mathcal{P}_n)$ which tends to work well
in practice (see for example \cite[Section~IV]{Fi22b}), and for
which it is possible
to derive approximation error bounds with a geometric convergence rate
of $\f{1}{n^{1/2}}$ based on the spiral geometry
(see Corollary~\ref{cor:spiral}).
The spiral is then reflected over the imaginary
axis to ensure that $\mathcal{P}_n$ is closed under complex conjugation.
Letting $m$ be the total number of poles in $\mathcal{P}_n$, the poles chosen
for this selection (as shown in Fig.~\ref{fig:spiral}) are given by
\begin{align}
  \begin{split}
  \theta_k &= 2 \sqrt{\pi k}, \\  r_k &= \sqrt{\f{k}{\f{m}{2}+1}}, \\
  p_k &= (r_k,\theta_k),  \\
  p_{-k} &= (r_k,-\theta_k)
  \label{eq:poles}    
  \end{split}
\end{align}
for $k \in \left\{1,...,\f{m}{2}\right\}$.

Theorem~\ref{thm:spiral} shows that the Archimedes spiral pole selection
of \eqref{eq:poles} yields a space-filling sequence of poles with
geometric convergence rate $\f{1}{n^{1/2}}$.

\begin{theorem}[Spiral Geometric Convergence Rate]\label{thm:spiral}
  For each integer $n \geq 2$, let $\p_n$ denote the selection of $(2n-2)$
  poles along an Archimedes spiral given by $p_k$ in \eqref{eq:poles} for
  $k \in \left\{-(n-1),...,-1,1,...,n-1\right\}$.
  Then $\{\p_n\}_{n=2}^\infty$ is a space-filling sequence of poles
  and has geometric convergence rate $\f{1}{n^{1/2}}$.
  In particular, for $S \in \f{1}{z}\h$ it satisfies
  \begin{align*}
    D(\p_n) &\leq \f{c_S}{n^{1/2}} \\
    c_S &= \sqrt{\pi}\left(1 + m_{\max}\right)
  \end{align*}
  for all positive integers $n \geq 2$.
\end{theorem}

Corollary~\ref{cor:spiral} provides a uniform convergence rate of
$\f{1}{n^{1/2}}$ for
the simple pole approximation using the Archimedes spiral pole selection
of \eqref{eq:poles} to any transfer function in $\f{1}{z}\h$.

\begin{corollary}[Spiral Approximation Convergence Rate]\label{cor:spiral}
  Consider the pole selection of Theorem~\ref{thm:spiral}.
  Then for any $S \in \f{1}{z}\h$ there exist constants
  $c_S = c_S(\mathcal{Q},G_{(q,j)}^*) > 0$,
  $c_S' = c_S'(\mathcal{Q},G_{(q,j)}^*) > 0$, and $N > 0$ such that
  for each integer $n \geq N$ there exist $\{G_p^n\}_{p \in \p_n}$ such that
  \begin{align*}
    \left|\left|\sum_{p \in \p_n} G_p^n \f{1}{z-p} -S\right|\right|_{\mathcal{H}_2}
    &\leq \f{c_S}{n^{1/2}} \\
   \left|\left|\sum_{p \in \p_n} G_p^n \f{1}{z-p} -S\right|\right|_{\mathcal{H}_\infty}
    &\leq \f{c_S'}{n^{1/2}}.
  \end{align*}
\end{corollary}

Note that $N$ in Corollary~\ref{cor:spiral} is chosen only to satisfy the
conditions in the remark after Theorem~\ref{thm:conv}, and so is typically
small in practice (i.e., $N \approx 4$).

For comparison, Galerkin-type approximations based on Lorentz series,
such as FIR, often have exponential convergence rates of the form
$\f{1}{\rho^n}$ for some $\rho \in (0,1)$, 
which is a faster convergence rate than in Corollary~\ref{cor:spiral}.
However, this rate $\rho$ typically depends on the decay rate of the optimal
transfer function, which may be very slow in practice, especially for systems
with large separation of time scales.
Additionally, for this convergence rate to apply, it may require $n$
sufficiently large such that $\rho^n$ has decayed sufficiently
(see \cite[Theorem~4.7]{An19}),
and so may not apply in practice
(see for example \cite[Section~IV]{Fi22b}).
In contrast, the convergence rate in Corollary~\ref{cor:spiral} is uniform
over all $S \in \f{1}{z}\h$, and so does not depend on the decay rate of the
optimal transfer function, and therefore tends to work well even for systems
with large separation of time scales (see for example \cite[Section IV]{Fi22b}).
In addition, if some optimal poles can be included in SPA due to prior
knowledge (see Section~\ref{sec:intro}), this will typically have the effect
of decreasing the constants $c_S$ and $c_S'$, as well as reducing $D(\p_n)$,
resulting in faster convergence.
For approximations based on Lorentz series, it is unclear how such prior
information could be included to improve the convergence rate.
Furthermore, when the optimal transfer function being approximated represents
a closed-loop response, such as for system level synthesis and input-output
parameterization, Lorentz series approximations are generally unable to
include poles of the plant which are stable but uncontrollabe.
In such cases, the control design may be infeasible, so additional slack
variables and constraint violations would have to be introduced
\cite[Eq.~4.36]{An19}, further
reducing the convergence rate to the optimal transfer function
\cite[Theorem~4.7]{An19}.
For SPA, the poles of the plant can be automatically incorporated, so
feasibility of the control design is guaranteed for stabilizable plants,
and the reductions in convergence rate resulting from introducing slack
variables can be avoided.

\section{Proofs}\label{sec:proof}

\subsection{Proof of Theorem~\ref{thm:approx}}

The proof of Theorem~\ref{thm:approx} proceeds by bounding the error in
the simple pole approximation in terms of the distance between the
approximating poles and the repeated poles they are approximating.
The next two lemmas provide useful identities for developing approximation
error bounds for a single repeated pole in the SISO case.

\begin{lemma}\label{lem:first}
  For any positive integer $m$, let $p_1,...,p_m, q \in \mathbb{D}$, and let
  $I_m = \{1, 2, \hdots, m\}$.  Then
  \begin{align*}
    \left|q^m - \prod_{i=1}^m p_i \right|
    \leq \sum_{k=1}^m \sum_{\substack{S \subset I_m \\ |S| = k}} |q|^{m-k}
    \prod_{i \in S} |p_i-q|.
  \end{align*}
\end{lemma}

\begin{proof}[Proof of Lemma~\ref{lem:first}]
  We compute
  \begin{align*}
    \prod_{i=1}^m p_i = \prod_{i=1}^m ((p_i-q) + q)
    \overset{\substack{\text{distributive} \\ \text{property}}}{=}
    \sum_{k=0}^m \sum_{\substack{S \subset I_m \\ |S| = k}} q^{m-k}
    \prod_{i \in S}(p_i-q)
  \end{align*}
  which implies that
  \begin{align*}
    \left|q^m - \prod_{i=1}^m p_i\right|
    &
    \overset{\substack{\text{above} \\ \text{identity}}}{=}
    \left|q^m -\sum_{k=0}^m \sum_{\substack{S \subset I_m \\ |S| = k}} q^{m-k}
    \prod_{i \in S}(p_i-q)\right| \\
    &
    \overset{\substack{\text{canceling} \\ q^m}}{=}
    \left|-\sum_{k=1}^m \sum_{\substack{S \subset I_m \\ |S| = k}} q^{m-k}
    \prod_{i \in S}(p_i-q)\right| \\
    &
    \overset{\substack{\text{triangle} \\ \text{inequality}}}{\leq}
    \sum_{k=1}^m \sum_{\substack{S \subset I_m \\ |S| = k}} |q|^{m-k}
    \prod_{i \in S} |p_i-q|
  \end{align*}
  by the triangle inequality.
\end{proof}

\begin{lemma}\label{lem:second}
  For any positive integer $m$, let $p_1,...,p_m, q \in \d$, and
  let $z \in \c$.
  Then there exist constants $c_1,...,c_m$ such that
  \begin{align}
    \begin{split}
    &\left| \sum_{i=1}^m c_i \f{1}{z-p_i} - \f{1}{(z-q)^m} \right| \\
    & \quad \quad
    \leq  \f{\sum\limits_{k=1}^m \sum\limits_{\substack{S \subset I_m \\ |S| = k}} 
      \sum\limits_{i=1}^k \sum\limits_{\substack{T \subset S \\ |T| = i}} |q|^{k-i}
      \prod\limits_{j \in T}|p_j-q|}
      {d(q,\c)^m \prod\limits_{i=1}^m d(p_i,\c)}.
    \end{split}
    \label{eq:sums}
  \end{align}
\end{lemma}

\begin{proof}[Proof of Lemma~\ref{lem:second}]
  First we show there exist constants $\{c_i\}_{i=1}^m$ such that
  the following partial fraction decomposition holds:
  \begin{align}
    \sum\nolimits_{i=1}^m c_i \f{1}{z-p_i} = \f{1}{\prod_{i=1}^m (z-p_i)}.
    \label{eq:cci}
  \end{align}
  Multiplying both sides by the product $\prod_{i=1}^m (z-p_i)$
  and evaluating at $z = p_i$ gives
  \begin{align}
    c_i &= \f{1}{\prod_{\substack{j=1 \\ j \neq i}}^m (p_i-p_j)}
    \label{eq:ci}
  \end{align}
  for all $i \in I_m$, which satisfy \eqref{eq:cci}.
So, after choosing constants $\{c_i\}_{i=1}^m$ by \eqref{eq:ci}, to prove the
claim it suffices to show that
$\left|\f{1}{\prod_{i=1}^m (z-p_i)} - \f{1}{(z-q)^m} \right|$
satisfies the inequality of Lemma~\ref{lem:second}.
We compute
\begin{align}
  \f{1}{\prod\limits_{i=1}^m (z-p_i)} - \f{1}{(z-q)^m}
  &=
  \f{(z-q)^m - \prod\limits_{i=1}^m (z-p_i)}{(z-q)^m\prod\limits_{i=1}^m (z-p_i)}.
  \label{eq:obv}
\end{align}
Then 
\begin{align*}
  &(z-q)^m - \prod\limits_{i=1}^m (z-p_i) \\
  &
  \overset{\substack{\text{binomial} \\ \text{theorem}}}{=}
  \sum_{k=0}^m {m \choose k}
  z^{m-k}(-q)^k - \sum_{k=0}^m \sum_{\substack{S \subset I_m \\ |S| =
      k}} z^{m-k} \prod_{j \in S} (-p_j)
  \\ &
  \overset{\substack{\text{canceling} \\ z^m}}{=}  
  \sum_{k=1}^m z^{m-k} {m \choose k} (-q)^k
  - \sum_{k=1}^m z^{m-k} \sum_{\substack{S \subset I_m \\ |S| = k}} 
  \prod_{j \in S} (-p_j)
  \\ &
  \overset{\substack{\text{difference} \\ \text{of sums}}}{=}
  \sum_{k=1}^m
  z^{m-k} \sum_{\substack{S \subset I_m \\ |S| = k}} \left((-q)^k -
  \prod_{j \in S} (-p_j) \right)
\end{align*}
where the last equality follows since
the number of terms in $\sum_{\substack{S \subset I_m \\ |S| = k}}$ is the
number of ways to select $k$ poles out of $m$ poles, which is
${m \choose k}$.
Then, recalling that $z \in \c$ so $|z| = 1$, and
applying Lemma~\ref{lem:first} to poles $-q$ and $\{-p_i\}_{i=1}^k$ for each
set $S$ in the sum yields
\begin{align}
  &\left|(z-q)^m - \prod\limits_{i=1}^m (z-p_i)\right| \nonumber \\
  &
  \overset{\substack{\text{triangle} \\ \text{inequality}}}{\leq}
  \sum_{k=1}^m |z|^{m-k} \sum_{\substack{S \subset I_m \\ |S| = k}} \left|(-q)^k
  - \prod_{j \in S} (-p_j) \right| \nonumber \\
  &
  \overset{|z| \leq 1}{\leq}
  \sum_{k=1}^m \sum_{\substack{S \subset I_m \\ |S| = k}} \left|(-q)^k
  - \prod_{j \in S} (-p_j) \right| \nonumber \\
  &
  \overset{\text{Lemma}~\ref{lem:first}}{\leq}  
  \sum_{k=1}^m \sum_{\substack{S \subset I_m \\ |S| = k}} 
  \sum_{i=1}^k \sum_{\substack{T \subset S \\ |T| = i}} |q|^{k-i}
  \prod_{j \in T}|p_j-q|.
  \label{eq:top}
\end{align}
Furthermore, since $z \in \c$
\begin{align}
  \left|(z-q)^m\prod_{i=1}^m (z-p_i)\right|
  \geq d(q,\c)^m \prod_{i=1}^m d(p_i,\c). \label{eq:bottom}
\end{align}
Thus, taking the absolute value of \eqref{eq:obv} and combining with
\eqref{eq:top}-\eqref{eq:bottom}
we obtain the desired bound.
\end{proof}

Corollary~\ref{cor:siso} provides an approximation error bound for a single
repeated pole in the SISO case.

\begin{corollary}\label{cor:siso}
 Let $p_1,...,p_m \in \overline{B}_r$ for some $r \in (0,1)$,
  $q \in \d$, 
  $z \in \c$, and $\p = \cup_{i=1}^m p_i$.
  Let $\hat{d}(q) = \max_i |p_i-q|$
  and assume that $\hat{d}(q) < 1$.
  Let $\bullet \in \{2,\infty\}$. 
  Then there exist constants $c_1,...,c_m$ such that
  \begin{align*}
    &\left|\left| \sum\nolimits_{i=1}^m c_i \f{1}{z-p_i} - \f{1}{(z-q)^m}
    \right|\right|_{H_\bullet} \\
    & \quad \quad
    \leq \left(\f{(|q|+2)^m-(|q|+1)^m}
           {d(q,\c)^m(1-r)^m}\right)\hat{d}(q).
  \end{align*}
\end{corollary}

\begin{proof}[Proof of Corollary~\ref{cor:siso}]
  To prove the claim, we upper bound the right-hand side of \eqref{eq:sums}.
  We compute 
  \begin{align}
    &\sum\limits_{k=1}^m \sum\limits_{\substack{S \subset I_m \\ |S| = k}} 
      \sum\limits_{i=1}^k \sum\limits_{\substack{T \subset S \\ |T| = i}} |q|^{k-i}
      \prod\limits_{j \in T}|p_j-q| \nonumber \\
    &
    \overset{\substack{\text{bound} \\ \text{by } \hat{d}(q)}}{\leq}
    \sum\limits_{k=1}^m \sum\limits_{\substack{S \subset I_m \\ |S| = k}} 
    \sum\limits_{i=1}^k \sum\limits_{\substack{T \subset S \\ |T| = i}}
    |q|^{k-i}\hat{d}(q)^i
    \nonumber \\
    &
    \overset{\substack{\text{combining} \\ \text{equal terms}}}{=}
    \sum\limits_{k=1}^m \sum\limits_{\substack{S \subset I_m \\ |S| = k}} 
    \sum\limits_{i=1}^k {k \choose i} |q|^{k-i}\hat{d}(q)^i
    \nonumber \\
    &
    \overset{\substack{\text{binomial} \\ \text{theorem}}}{=}
    \sum\limits_{k=1}^m \sum\limits_{\substack{S \subset I_m \\ |S| = k}}
    \left((|q|+\hat{d}(q))^k - |q|^k\right)
    \nonumber \\
    &
    \overset{\substack{\text{combining} \\ \text{equal terms}}}{=}
    \sum\limits_{k=1}^m {m \choose k} \left((|q|+\hat{d}(q))^k - |q|^k\right)
    \nonumber \\
    &
    \overset{\substack{\text{binomial} \\ \text{theorem}}}{=}
    \left((|q| + \hat{d}(q) + 1)^m - 1\right) - \left((|q|+1)^m - 1\right)
    \nonumber \\
    &
    \overset{\substack{\text{canceling} \\ 1}}{=}
    (|q|+1+\hat{d}(q))^m - (|q|+1)^m \nonumber \\
    &
    \overset{\substack{\text{binomial} \\ \text{theorem}}}{=}
    \sum_{k=0}^m {m \choose k} (|q|+1)^{m-k}\hat{d}(q)^k -  (|q|+1)^m \nonumber
    \\
    &
    \overset{\substack{\text{canceling} \\ (|q|+1)^m}}{=}
    \sum_{k=1}^m {m \choose k} (|q|+1)^{m-k}\hat{d}(q)^k \nonumber \\
    &
    \overset{\hat{d}(q)^k \leq \hat{d}(q)}{\leq}
    \left(\sum_{k=1}^m {m \choose k} (|q|+1)^{m-k}\right)\hat{d}(q)
    \nonumber\\
    &
    \overset{\substack{\text{binomial} \\ \text{theorem}}}{=}
    \left((|q|+2)^m - (|q|+1)^m\right)\hat{d}(q)
    \label{eq:K}
  \end{align}
  since $\hat{d}(q)^k \leq \hat{d}(q)$ because $\hat{d}(q) < 1$.
  Note that
  \begin{align*}
    d(q,\c)^m\prod_{i=1}^m d(p_i,\c) &\geq d(q,\c)^md(\p,\c)^m \\
    &\geq d(q,\c)^m(1-r)^m \\
  \end{align*}
  since $\p \subset \overline{B}_r$.
  Choose constants $\{c_i\}_{i=1}^m$ as in \eqref{eq:ci} in
    the proof of Lemma~\ref{lem:second}.
  Then applying the above inequalities to the result of Lemma~\ref{lem:second},
  we obtain
  \begin{align*}
    \left| \sum_{i=1}^m c_i \f{1}{z-p_i} - \f{1}{(z-q)^m} \right|
    &\leq \f{(|q|+2)^m - (|q|+1)^m}{d(q,\c)^m(1-r)^m} \hat{d}(q).
  \end{align*}
  Taking the supremum over $z \in \c$ yields
  \begin{align*}
    &\left|\left| \sum_{i=1}^m c_i \f{1}{z-p_i} - \f{1}{(z-q)^m}
    \right|\right|_{\mathcal{H}_\infty}
    \\&= \sup_{z \in \c}
    \left| \sum_{i=1}^m c_i \f{1}{z-p_i} - \f{1}{(z-q)^m} \right|
    \\&\leq \f{(|q|+2)^m - (|q|+1)^m}{d(q,\c)^m(1-r)^m} \hat{d}(q).
  \end{align*}
  Finally, for any 
  $T \in \f{1}{z}\h$ SISO,
  $||T||_{H_2} \leq ||T||_{\mathcal{H}_\infty}$ which, combined with the above
  inequality,
    yields the desired bound.
\end{proof}

Theorem~\ref{thm:approx} extends the approximation error bound to an arbitrary
number of (possibly repeated) poles and to the MIMO case.

\begin{proof}[Proof of Theorem~\ref{thm:approx}]
The proof begins by writing the partial fraction decomposition of $T$, and
for each pole $q$ of $T$ using the SISO approximation of
Corollary~\ref{cor:siso} to approximate each SISO term $\f{1}{(z-q)^j}$ with
the $j$ nearest poles in $\p$.
These are then combined with the MIMO coefficients in the partial fraction
decomposition of $T$ to yield the approximating transfer function
$\sum_{p \in \p} G_p \f{1}{z-p}$.
Care must be taken to ensure symmetry between approximations of complex
conjugate poles in $T$ so that the approximating transfer function has real
coefficients.
Next, using the SISO approximation error bounds of Corollary~\ref{cor:siso},
the main approximation error bounds of the corollary are derived.
Finally, it is shown that, by construction, the approximating transfer function
does indeed have real coefficients, and therefore belongs to $\f{1}{z}\h$.

We begin by writing the partial fraction decomposition of $T$ and constructing
the approximating transfer function.
Let $\mathcal{Q}_{\mathbb{R}} \subset \mathcal{Q}$ denote the real poles of
$T$, and let $\mathcal{Q}_{\mathbb{C}} \subset \mathcal{Q}$ denote the remaining
poles.
Since $T \in \f{1}{z}\h$, we can write its partial fraction decomposition
with matrix-valued coefficients $C_{(q,j)}$ as
\begin{align}
  T(z) = \sum_{q \in \mathcal{Q}_{\mathbb{R}}} \sum_{j=1}^{m_q} C_{(q,j)} \f{1}{(z-q)^j}
  + \sum_{q \in \mathcal{Q}_{\mathbb{C}}} \sum_{j=1}^{m_q} C_{(q,j)} \f{1}{(z-q)^j}
  \label{eq:Tpfd}
\end{align}
For each $q \in \mathcal{Q}$ and each $j \in I_{m_q}$,
let $\p(q,j) \subset \p(q)$ denote the $j$ closest poles in $\p$ to $q$
(or at least one choice in case this is not unique)
and choose constants $\{c_p^{(q,j)}\}_{p \in \p(q,j)}$ as in
Corollary~\ref{cor:siso} for approximating the pole $q$ with multiplicity $j$
by $\p(q,j)$.
Then, choose $\p(\overline{q},j) = \overline{\p(q,j)}$ (note that by
symmetry, since $\p$ and $\mathcal{Q}$ are both closed under complex
conjugation,
these are the $j$ closest poles in $\p$ to $\overline{q}$) and choose
constants $\{c_p^{(\overline{q},j)}\}_{p \in \p(\overline{q},j)}$ as in
Corollary~\ref{cor:siso} for approximating the pole $\overline{q}$ with
multiplicity $j$ by
$\p(\overline{q},j)$.
Note that in case $q$ is real, for each $j \in I_m$ this results in two
sets of constants for approximating $q$: one for $\p(q,j)$ and one for
$\p(\overline{q},j) = \overline{\p(q,j)}$ (where we have abused notation for
simplicity of presentation).
Then for each $p \not\in \p(q,j)$, let $c_p^{(q,j)} = 0$, and 
for each $p \not\in \p(\overline{q},j)$,
let $c_p^{(\overline{q},j)} = 0$ for notational convenience.
For each $q \in \mathcal{Q}$ and each $j \in I_{m_q}$,
let $\hat{d}(q,j) = \max_{p \in \p(q,j)} |p-q|$.
Then $\hat{d}(q,j) \leq \hat{d}(q) \leq D(\p)$.
By Corollary~\ref{cor:siso}, this implies that for each $q \in \mathcal{Q}$
and each $j \in I_{m_q}$
\begin{align}
  \begin{split}
    &\left|\left| \sum_{p \in \p} c_p^{(q,j)} \f{1}{z-p} - \f{1}{(z-q)^j}
    \right|\right|_{\mathcal{H}_\infty} \\
    & \quad \quad
    \overset{\text{Corollary}~\ref{cor:siso}}{\leq}
    \left(\f{(|q|+2)^j-(|q|+1)^j}
         {d(q,\c)^j(1-r)^j}\right) \hat{d}(q,j) \\
    & \quad \quad
    \overset{\hat{d}(q,j) \leq D(\p)}{\leq}
    \left(\f{(|q|+2)^j-(|q|+1)^j}
         {d(q,\c)^j(1-r)^j}\right) D(\p)
   \label{eq:cp}
  \end{split}
\end{align}
and similarly for $\overline{q}$.
Then for each $p \in \p$, define
\begin{align}
  \begin{split}
  G_p &= \sum_{q \in \mathcal{Q}_{\mathbb{R}}} \sum_{j=1}^{m_q}
    C_{(q,j)} \f{c_p^{(q,j)} + c_p^{(\overline{q},j)}}{2}
    + \sum_{q \in \mathcal{Q}_{\mathbb{C}}} \sum_{j=1}^{m_q} C_{(q,j)} c_p^{(q,j)}.
    \label{eq:gp}
    \end{split}
\end{align}
This completes the construction of the approximating transfer function
$\sum_{p \in \p} G_p \f{1}{z-p}$.

Next we show that the approximating transfer function satisfies the
desired approximation error bounds of \eqref{eq:mimo_approx}.
We compute 
\begin{align*}
  &\left|\left|
  \sum_{p \in \p} G_p \f{1}{z-p} - T
  \right|\right|_{\mathcal{H}_\infty} \\
  &
  \overset{\substack{\eqref{eq:Tpfd} \\ \eqref{eq:gp}}}{=}
  \left|\left|
  \sum_{q \in \mathcal{Q}_{\mathbb{R}}} \sum_{j=1}^{m_q} C_{(q,j)}
  \left(\sum_{p \in \p} \f{c_p^{(q,j)}+c_p^{(\overline{q},j)}}{2}
  \f{1}{z-p} -\f{1}{(z-q)^j} \right)\right.\right.\\
  & \left.\left.+ \sum_{q \in \mathcal{Q}_{\mathbb{C}}} \sum_{j=1}^{m_q} C_{(q,j)}
  \left(\sum_{p \in \p} c_p^{(q,j)} \f{1}{z-p} -\f{1}{(z-q)^j} \right)  
  \right|\right|_{\mathcal{H}_\infty} \\
  &
  \overset{\substack{\text{triangle} \\ \text{inequality}}}{\leq}
  \mkern-10mu
  \sum_{q \in \mathcal{Q}_{\mathbb{R}}} \sum_{j=1}^{m_q}
  ||C_{(q,j)}||_2 \f{1}{2}
  \left|\left|
  \sum_{p \in \p} c_p^{(q,j)}
  \f{1}{z-p} - \f{1}{(z-q)^j}
  \right|\right|_{\mathcal{H}_\infty} \\
  &+ \sum_{q \in \mathcal{Q}_{\mathbb{R}}} \sum_{j=1}^{m_q}
  ||C_{(q,j)}||_2 \f{1}{2}
  \left|\left|
   \sum_{p \in \p} c_p^{(\overline{q},j)}
  \f{1}{z-p} - \f{1}{(z-\overline{q})^j}
  \right|\right|_{\mathcal{H}_\infty} \\
  &+ \sum_{q \in \mathcal{Q}_{\mathbb{C}}} \sum_{j=1}^{m_q}
  ||C_{(q,j)}||_2
  \left|\left|
   \sum_{p \in \p} c_p^{(q,j)}
  \f{1}{z-p} - \f{1}{(z-q)^j}
  \right|\right|_{\mathcal{H}_\infty}
  \\
  &
  \overset{\eqref{eq:cp}}{\leq}
  \sum_{q \in \mathcal{Q}_{\mathbb{R}}} \sum_{j=1}^{m_q} ||C_{(q,j)}||_2 \f{1}{2}
  \f{(|q|+2)^j-(|q|+1)^j}{d(q,\c)^j(1-r)^j} D(\p) \\
  &+ \sum_{q \in \mathcal{Q}_{\mathbb{R}}} \sum_{j=1}^{m_q} ||C_{(q,j)}||_2 \f{1}{2}
  \f{(|q|+2)^j-(|q|+1)^j}{d(q,\c)^j(1-r)^j} D(\p) \\
  &+ \sum_{q \in \mathcal{Q}_{\mathbb{C}}} \sum_{j=1}^{m_q} ||C_{(q,j)}||_2
  \f{(|q|+2)^j-(|q|+1)^j}{d(q,\c)^j(1-r)^j} D(\p)
  \\
  &
  \overset{\substack{\text{combining} \\ \text{equal terms}}}{=}
  \sum_{q \in \mathcal{Q}} \sum_{j=1}^{m_q} ||C_{(q,j)}||_2
  \f{(|q|+2)^j-(|q|+1)^j}{d(q,\c)^j(1-r)^j} D(\p) \\
  &= K_\infty D(\p) \\
  &K_\infty = \sum_{q \in \mathcal{Q}} \sum_{j=1}^{m_q} ||C_{(q,j)}||_2
  \f{(|q|+2)^j-(|q|+1)^j}{d(q,\c)^j(1-r)^j}.
\end{align*}
This gives the result for the $\mathcal{H}_\infty$ case.
Let $s$ be the minimum of the dimensions of $T$,
and note that
\begin{align*}
  ||T||_{\mathcal{H}_2} \leq \sqrt{s} ||T||_{\mathcal{H}_\infty}
  \leq \sqrt{s} K_\infty D(\p)
  = K_2 D(\p)
\end{align*}
where $K_2 = \sqrt{s} K_\infty$.  This completes the result for the
$\mathcal{H}_2$ case.

Finally, we show that $\sum_{p \in \p} G_p \f{1}{z-p} \in \f{1}{z} \h$.
Clearly $\sum_{p \in \p} G_p \f{1}{z-p}$ is strictly proper, rational,
and stable, so it suffices to show it has real coefficients.
Since $T$ has real coefficients it satisfies
$T(z) = \overline{T(\overline{z})}$ which implies, by matching coefficients
in the partial fraction decomposition and since
$\p(\overline{q},j) = \overline{\p(q,j)}$ for all $q \in \mathcal{Q}$
and $j \in I_m$
\begin{align}
  C_{(\overline{q},j)} = \overline{C_{(q,j)}}
  \label{eq:cq1}
\end{align}
for all $q \in \mathcal{Q}_{\mathbb{C}}$ and $j \in I_{m_q}$, and that
\begin{align}
  \overline{C_{(q,j)}} = C_{(q,j)}
  \label{eq:cq2}
\end{align}
for all $q \in \mathcal{Q}_{\mathbb{R}}$ and $j \in I_{m_q}$.
Furthermore, using the definition of $\{c_p^{(q,j)}\}_{p \in \p(q,j)}$ from
\eqref{eq:ci} and since
$\p(\overline{q},j) = \overline{\p(q,j)}$ for all $q \in \mathcal{Q}$
and $j \in I_m$, it is straightforward to verify that
\begin{align}
  c_{\overline{p}}^{(\overline{q},j)} = \overline{c_p^{(q,j)}}
  \label{eq:cpq}
\end{align}
for all $q \in \mathcal{Q}$ and $j \in I_{m_q}$.
It is also straightforward to verify 
that for any complex-valued
matrix $M$ and pole $p \in \d$, 
$M \f{1}{z-p} + \overline{M} \f{1}{z-\overline{p}}$ can
be expressed as a rational transfer function matrix with real coefficients
(we will refer to this later as fact (a)).
Let $\mathcal{Q}_+ \subset \mathcal{Q}_{\mathbb{C}}$ be the subset of poles with
nonnegative imaginary component, and similarly for $\p_+ \subset \p$.
Then 
\begin{align*}
  &\sum_{p \in \p} G_p \f{1}{z-p} \\
  &
  \overset{\eqref{eq:gp}}{=}
  \sum_{p \in \p_+} \sum_{q \in \mathcal{Q}_{\mathbb{R}}} \sum_{j=1}^{m_q} C_{(q,j)}
  \left(\f{c_p^{(q,j)} + c_p^{(\overline{q},j)}}{2} \f{1}{z-p} \right. \\
  & \quad \quad \quad \quad \quad \quad \quad \quad \quad \quad
  \left.+ \f{c_{\overline{p}}^{(q,j)} + c_{\overline{p}}^{(\overline{q},j)}}{2}
  \f{1}{z-\overline{p}} \right)\\
    &+ \sum_{p \in \p}
    \sum_{q \in \mathcal{Q}_+} \sum_{j=1}^{m_q} \left(C_{(q,j)} c_p^{(q,j)} \f{1}{z-p}
    + C_{(\overline{q},j)} c_p^{(\overline{q},j)} \f{1}{z-\overline{p}}\right) \\
  &
  \overset{\substack{\text{regrouping} \\ \text{terms}}}{=}
    \sum_{p \in \p_+} \sum_{q \in \mathcal{Q}_{\mathbb{R}}} \sum_{j=1}^{m_q} C_{(q,j)}
    \left(\f{c_p^{(q,j)}}{2} \f{1}{z-p} + \f{c_{\overline{p}}^{(\overline{q},j)}}{2}
    \f{1}{z-\overline{p}}
    \right. \\
  & \quad \quad \quad \quad \quad \quad \quad \quad \quad \quad
  \left.+ \f{c_p^{(\overline{q},j)}}{2} \f{1}{z-p}
  + \f{c_{\overline{p}}^{(q,j)}}{2}
  \f{1}{z-\overline{p}} \right)\\
    &+ \sum_{p \in \p}
    \sum_{q \in \mathcal{Q}_+} \sum_{j=1}^{m_q} \left(C_{(q,j)} c_p^{(q,j)} \f{1}{z-p}
    + C_{(\overline{q},j)} c_p^{(\overline{q},j)} \f{1}{z-\overline{p}}\right) \\
  &
  \overset{\substack{\eqref{eq:cq1} \\ \eqref{eq:cpq}}}{=}      
    \sum_{p \in \p_+} \sum_{q \in \mathcal{Q}_{\mathbb{R}}} \sum_{j=1}^{m_q} C_{(q,j)}
    \left(\f{c_p^{(q,j)}}{2} \f{1}{z-p} + \f{\overline{c_p^{(q,j)}}}{2}
    \f{1}{z-\overline{p}}
    \right. \\
  & \quad \quad \quad \quad \quad \quad \quad \quad \quad \quad
  \left.+ \f{c_p^{(\overline{q},j)}}{2} \f{1}{z-p}
  + \f{\overline{c_p^{(\overline{q},j)}}}{2}
  \f{1}{z-\overline{p}} \right)\\
    &+ \sum_{p \in \p}
    \sum_{q \in \mathcal{Q}_+} \sum_{j=1}^{m_q} \left(C_{(q,j)} c_p^{(q,j)} \f{1}{z-p}
    + \overline{C_{(q,j)} c_p^{(q,j)}} \f{1}{z-\overline{p}}\right)
\end{align*}
which, by fact (a) above and \eqref{eq:cq2},
is a sum of rational transfer functions with real coefficients,
hence it is rational with real coefficients.
\end{proof}

\subsection{Proofs of Theorem~\ref{thm:dens} and Theorem~\ref{thm:conv}}

This section will first use the approximation error bounds of
Theorem~\ref{thm:approx} to prove that the space of simple pole approximations
converges to the Hardy space $\f{1}{z}\h$ in Theorem~\ref{thm:dens}.
Then, Theorem~\ref{thm:approx} is combined with the notion of geometric
convergence rate to provide a uniform convergence rate in
Theorem~\ref{thm:conv} for the simple pole approximation to any transfer
function in $\f{1}{z}\h$.
The first step is to show the equivalence of space-filling and exhibiting
geometric convergence pole sequences in Lemma~\ref{lem:equiv}.

\begin{proof}[Proof of Lemma~\ref{lem:equiv}]
  To show that geometric
  convergence implies space-filling, a finite set
  of poles with small Hausdorff distance from the unit disk is chosen.
  Then, selecting a transfer function whose poles contain this pole selection,
  and using that the pole sequence exhibits geometric convergence,
  it is possible to choose a pole in $\p_n$ close to each pole in the
  pole selection.
  This will result in the Hausdorff distance between $\p_n$ and the pole
  selection being small, and hence small Hausdorff distance between $\p_n$
  and the unit disk, which will lead to the space-filling property.
  Next, assuming that the sequence of poles is space-filling, to show that
  it exhibits geometric convergence we begin by selecting an arbitrary transfer
  function.  Then, letting its maximum multiplicity be $m_{\max}$, $m_{\max}$
  poles near each of its poles are chosen in the unit disk.  Using the
  space-filling property, poles in $\p_n$ are found close to each of these
  poles in the unit disk, and hence close to each pole of the transfer function,
  establishing geometric convergence.
  
  Let $\{\p_n\}_{n=1}^\infty$ be a sequence of poles.    
  First suppose that $\{\p_n\}_{n=1}^\infty$ exhibits geometric convergence.
  To show that $\p_n \to \dc$, it suffices to show for every $\epsilon > 0$
  there exists $N$ such that $n \geq N$ implies that
  $d_H(\p_n,\dc) \leq \epsilon$.
  So, let $\epsilon > 0$.
  Choose any finite set of poles $\mathcal{Q} \subset \d$ such that
  $\mathcal{Q}_{\f{\epsilon}{2}} \supset \dc$, and note that this implies that
  $\cup_{q \in \mathcal{Q}} \{q\}_{\f{\epsilon}{2}} \supset \dc$.
  Choose $S \in \f{1}{z}\h$ which has simple poles at $\mathcal{Q}$
  (and possibly additional poles as well, such as the complex conjugate
  of $\mathcal{Q}$).
  Since $\{\p_n\}_{n=1}^\infty$ exhibits geometric convergence,
  there exists $N$ such that $n \geq N$ implies that
  $D(\p_n) < \f{\epsilon}{2}$.
  In particular, this implies that for any $n \geq N$ and every
  $q \in \mathcal{Q}$, there exists $p^{(q)} \in \p_n$ such that
  $|p^{(q)}-q| < \f{\epsilon}{2}$, which implies that
  $\{p^{(q)}\}_\epsilon \supset \{q\}_{\f{\epsilon}{2}}$.
  Thus, for every $n \geq N$,
  \begin{align*}
    (\p_n)_\epsilon \supset \cup_{q \in \mathcal{Q}} \{p^{(q)}\}_\epsilon
    \supset \cup_{q \in \mathcal{Q}} \{q\}_{\f{\epsilon}{2}} \supset \dc.
  \end{align*}
  As $\p_n \subset \dc$ for all $n$, this implies that for $n \geq N$,
  $d_H(\p_n,\dc) \leq \epsilon$.
  So, $\{\p_n\}_{n=1}^\infty$ is space-filling.

  Next suppose that $\{\p_n\}_{n=1}^\infty$ is space-filling.
  Let $S \in \f{1}{z}\h$.
  Let $\mathcal{Q}$ be the poles of $S$, and let
  $m_{\max} = \max_{q \in \mathcal{Q}} m_q$ be their maximum multiplicity, which
  is finite since $S$ is rational.
  To show that $\lim_{n \to \infty} D(\p_n) = 0$, it suffices to show that for
  every $\epsilon > 0$ there exists $N$ such that $n \geq N$ implies that
  $D(\p_n) \leq \epsilon$.
  So, let $\epsilon > 0$.
  For each $q \in \mathcal{Q}$, choose $q^{(1)},...,q^{(m_{\max})} \subset \d$
  such that for all $i, j \in \{1,...,m_{\max}\}$ with $j \neq i$,
  $|q-q^{(i)}| \leq \epsilon$ and $|q^{(i)}-q^{(j)}| > \f{\epsilon}{2 m_{\max}}$
  (one way to do this is to choose all of the $q^{(i)}$ along a line passing
  through $q$ such that $|q^{(i)}-q^{(i+1)}| = \f{\epsilon}{m_{\max}}$
  and $|q-q^{(i)}| = i \f{\epsilon}{m_{\max}}$ for all $i \in \{1,...,m_{\max}\}$).
  Since $\{\p_n\}_{n=1}^\infty$ is space-filling, there exists $N$ such that
  $n \geq N$ implies that $d_H(\p_n,\dc) < \f{\epsilon}{2 m_{\max}}$.
  Therefore, for each $n \geq N$, $q \in \mathcal{Q}$, and
  $i \in \{1,...,m_{\max}\}$ there exists $p^{(q,i)}_n \in \p_n$ such that
  $|q^{(i)}-p^{(q,i)}_n| < \f{\epsilon}{2 m_{\max}}$.
  By the choice of the $\{q^{(i)}\}_{i=1}^{m_{\max}}$, this implies that
  $\{p^{(q,i)}_n\}_{i=1}^{m_{\max}} \subset \p_n$ consists of $m_{\max}$ distinct
  poles, all of which satisfy $|p^{(q,i)}_n-q| \leq \epsilon$.
  Thus,
  \begin{align*}
    \hat{d}(q) = \max_{p \in \p_n(q)} |p-q| \leq
    \max_{i=1}^{m_{\max}} |p^{(q,i)}_n-q| \leq \epsilon
  \end{align*}
  since $\p_n(q)$ are the $m_q$ closest poles in $\p_n$ to $q$.
  As $q \in \mathcal{Q}$ was arbitrary, this implies that for all $n \geq N$,
  \begin{align*}
    D(\p_n) = \max_{q \in \mathcal{Q}} \hat{d}(q) \leq \epsilon.
  \end{align*}
  So, $\{\p_n\}_{n=1}^\infty$ exhibits geometric convergence.
\end{proof}

Now we are ready to prove Theorem~\ref{thm:dens}.

\begin{proof}[Proof of Theorem~\ref{thm:dens}]
  The proof begins by fixing an arbitary transfer function in $\f{1}{z}\h$,
  and uses Theorem~\ref{thm:approx} to obtain approximation error bounds
  using the simple pole approximation with poles $\p_n$ for each $n$.
  Some initial work is required to ensure the assumptions of
  Theorem~\ref{thm:approx} are met uniformly by $\p_n$ for all $n$
  sufficiently large.
  As the sequence of poles is space filling, by Lemma~\ref{lem:equiv} it
  exhibits geometric convergence, so taking the limit of the approximation
  error bounds as $n \to \infty$ implies convergence of the simple pole
  approximations to the desired transfer function.
  This implies that this transfer function is contained in the
  $\liminf_{n \to \infty}$ of the simple pole
  approximation spaces and, since it was arbitrary, that this
  $\liminf_{n \to \infty}$ is equal to $\f{1}{z}\h$.
  As the
  $\liminf_{n \to \infty}$ is contained in the $\limsup_{n \to \infty}$ is
  contained in $\f{1}{z}\h$,
  this completes the proof.
  
  Let $S \in \f{1}{z}\h$ and let $\mathcal{Q}$ be the poles of $S$.
  Let $r = \f{1}{2}\left(1+\max_{q \in \mathcal{Q}} |q|\right) \in (0,1)$.
  As $\{\p_n\}_{n=1}^\infty$ is space-filling, by Lemma~\ref{lem:equiv} it
  exhibits geometric convergence, so $\lim_{n \to \infty} D(\p_n) = 0$.
  Thus, there exists $N$ such that $n \geq N$ implies that
  $D(\p_n) < 1$ and
  $D(\p_n) < \f{1}{2}\left(1-\max_{q \in \mathcal{Q}} |q|\right)$.
  This implies that for $n \geq N$, the poles from $\p_n$ used to approximate
  $\mathcal{Q}$ in the simple pole approximation of Theorem~\ref{thm:approx}
  all lie within
  \begin{align*}
    \overline{B}_{\max_{q \in \mathcal{Q}} |q| + D(\p_n)}
    \subset \overline{B}_{\f{1}{2}\left(1 + \max_{q \in \mathcal{Q}} |q|\right)}
    = \overline{B}_r.
  \end{align*}
  Furthermore, since $\{\p_n\}_{n=1}^\infty$ is space-filling,
  increasing $N$ further if necessary implies that for $n \geq N$,
  $\p_n \cap \overline{B}_r$ contains at least $m_{\max}$ poles.
  Thus, for every $n \geq N$, $\p_n \cap \overline{B}_r$ is contained in
  $\overline{B}_r$ with $r \in (0,1)$ (Assumption A1),
  contains at least $m_{\max}$ poles
  (Assumption A2), is closed under complex conjugation (Assumption A3),
  and satisfies $D(\p_n) < 1$ (Assumption A4).
  So, by Theorem~\ref{thm:approx}, there exist constants $c_S, c_S' > 0$
  such that for every $n \geq N$, there exist coefficients
  $\{G_p^n\}_{p \in \p_n}$ such that
  \begin{align*}
    \left|\left| \sum_{p \in \p_n} G_p^n \f{1}{z-p} - S
    \right|\right|_{\mathcal{H}_2}
    &\leq c_S D(\p_n), \\
    \left|\left| \sum_{p \in \p_n} G_p^n \f{1}{z-p} - S
    \right|\right|_{\mathcal{H}_\infty}
    &\leq c_S' D(\p_n).
  \end{align*}
  Thus, taking the limit as $n \to \infty$ implies
  \begin{align*}
    \lim_{n \to \infty}  \left|\left| \sum_{p \in \p_n} G_p^n \f{1}{z-p} - S
    \right|\right|_{\mathcal{H}_2} &\leq c_S \lim_{n \to \infty} D(\p_n) = 0, \\
    \lim_{n \to \infty} \left|\left| \sum_{p \in \p_n} G_p^n \f{1}{z-p} - S
    \right|\right|_{\mathcal{H}_\infty}
    &\leq c_S' \lim_{n \to \infty} D(\p_n) = 0.
  \end{align*}
  Therefore, $\lim_{n \to \infty} \sum_{p \in \p_n} G_p^n \f{1}{z-p} = S$
  with respect to the metrics induced by both the $\mathcal{H}_2$ and
  $\mathcal{H}_\infty$ norms.
  As $\sum_{p \in \p_n} G_p^n \f{1}{z-p} \in \mathcal{A}_n$ for all $n \geq N$,
  and $\lim_{n \to \infty} \sum_{p \in \p_n} G_p^n \f{1}{z-p} = S$,
  $S \in \liminf_{n \to \infty} \mathcal{A}_n$.
  As $S \in \f{1}{z}\h$ was arbitrary,
  $\liminf_{n \to \infty} \mathcal{A}_n = \f{1}{z}\h$.
  By definition,
  \begin{align*}
    \f{1}{z}\h = \liminf_{n \to \infty} \mathcal{A}_n \subset
    \limsup_{n \to \infty} \mathcal{A}_n \subset \f{1}{z}\h
  \end{align*}
  so
  \begin{align*}
    \f{1}{z}\h = \liminf_{n \to \infty} \mathcal{A}_n = \limsup_{n \to \infty}
    \mathcal{A}_n = \lim_{n \to \infty} \mathcal{A}_n.
  \end{align*}
\end{proof}

Next we prove Theorem~\ref{thm:conv}.

\begin{proof}[Proof of Theorem~\ref{thm:conv}]
  The proof proceeds by combining the approximation error bounds of
  Theorem~\ref{thm:approx} with the definition of geometric convergence rate.
  Some initial work is required to ensure that the assumptions of
  Theorem~\ref{thm:approx} are met uniformly by $\p_n$ for all $n$
  sufficiently large.
  
  Let $S \in \f{1}{z}\h$, let $\mathcal{Q}$ be the poles of $S$, and let
  $m_{\max}$ be their maximum multiplicity.
  Let $r_0 = \f{1}{2}\left(1 + \max_{q \in \mathcal{Q}} |q|\right)$.
  Since $\{\p_n\}_{n=1}^\infty$ is a sequence of poles with
  geometric convergence rate $\f{1}{n^k}$, it exhibits geometric convergence.
  So, there exists $N'$ such that $n \geq N'$ implies that
  $D(\p_n) < \f{1}{2}\left(1 - \max_{q \in \mathcal{Q}} |q| \right)$.
  Thus, for $n \geq N'$, the poles from $\p_n$ used to approximate
  $\mathcal{Q}$ in the simple pole approximation of Theorem~\ref{thm:approx}
  all lie within
  \begin{align*}
    \overline{B}_{\max_{q \in \mathcal{Q}} |q| + D(\p_n)}
    \subset \overline{B}_{\f{1}{2}\left(1 + \max_{q \in \mathcal{Q}} |q|\right)}
    = \overline{B}_{r_0}.
  \end{align*}
  For $n \in \{1,...,N'-1\}$, let $r_n = \max_{p \in \p_n} |p|$, so
  $\p_n \cap \overline{B}_{r_n} = \p_n$.
  Let $r = \max_{n=0}^{N'-1} r_n$.
  Then for all $n$, the poles from $\p_n$ used to approximate $\mathcal{Q}$
  in the simple pole approximation of Theorem~\ref{thm:approx} all lie within
  $\p_n \cap \overline{B}_r$ (Assumption A1).
  Furthermore, $\p_n \cap \overline{B}_r$ is closed under complex conjugation
  for all $n$ (Assumption A3).
  As $\{\p_n\}_{n=1}^\infty$ exhibits geometric convergence, there exists $N$
  such that $n \geq N$ implies that $D(\p_n) < 1$ (Assumption A4).
  As $\{\p_n\}_{n=1}^\infty$ exhibits geometric convergence, by
  Lemma~\ref{lem:equiv} it is space-filling.
  So, increasing $N$ if necessary implies that for $n \geq N$,
  $|\p_n \cap \overline{B}_r| \geq m_{\max}$ (Assumption A2).
  Note that $N$ is chosen just to ensure that $D(\p_n) < 1$
  and that $\p_n$ has at least $m_{\max}$ poles.
  As Assumptions A1-A4 are satisfied for $\p_n \cap \overline{B}_r$
  with $n \geq N$, by Theorem~\ref{thm:approx} there exist
  constants $\hat{c}_S, \hat{c}_S' > 0$
  such that for every $n \geq N$, there exist coefficients
  $\{G_p^n\}_{p \in \p_n}$ such that
  \begin{align*}
    \left|\left| \sum_{p \in \p_n} G_p^n \f{1}{z-p} - S
    \right|\right|_{\mathcal{H}_2}
    &\leq \hat{c}_S D(\p_n), \\
    \left|\left| \sum_{p \in \p_n} G_p^n \f{1}{z-p} - S
    \right|\right|_{\mathcal{H}_\infty}
    &\leq \hat{c}_S' D(\p_n).
  \end{align*}
  As $\{\p_n\}_{n=1}^\infty$ is a sequence of poles with
  geometric convergence rate $\f{1}{n^k}$, there exists a constant
  $\tilde{c}_S > 0$
  such that for all $n$,
  \begin{align*}
    D(\p_n) \leq \f{\tilde{c}_S}{n^k}.
  \end{align*}
  Combining this with the approximation error bounds above yields
  \begin{align*}
  \left|\left| \sum_{p \in \p_n} G_p^n \f{1}{z-p} - S
    \right|\right|_{\mathcal{H}_2}
    &\leq \f{c_S}{n^k} \\
    \left|\left| \sum_{p \in \p_n} G_p^n \f{1}{z-p} - S
    \right|\right|_{\mathcal{H}_\infty}
    &\leq \f{c_S'}{n^k} \\
    c_S &= \hat{c}_S \tilde{c}_S \\
    c_S' &= \hat{c}_S' \tilde{c}_S
  \end{align*}
  for all $n \geq N$.
\end{proof}

\subsection{Proofs of Theorem~\ref{thm:spiral} and Corollary~\ref{cor:spiral}}

This section will prove Theorem~\ref{thm:spiral} and Corollary~\ref{cor:spiral}.
To achieve an approximately uniform selection of poles over the unit disk,
the basic idea is to choose a spiral whose windings are equally spaced apart,
and then to select poles along this spiral such that the distance between
any two successive poles is close to but less than the distance between the
windings.
Such a selection is shown in Fig.~\ref{fig:spiral} and first appeared in
\cite[Section 5.1]{Ar15}, although we are not aware of any prior uses for
control design.
Note that as the number of poles increases, the spiral also changes so that
the distance between windings (and, hence, between successive poles)
converges towards zero.
The main challenge is to show that the chosen pole selection along this spiral
does in fact possess these geometric properties.
Then, these will be used to show that as the number of poles increases,
$D(\p)$ converges to zero at the same
rate at which the distance between windings converges to zero.

Lemma~\ref{lem:arch} and Corollary~\ref{cor:dist} provide basic geometric facts
about the Archimedes spiral as parameterized in \eqref{eq:poles}.
Lemma~\ref{lem:arch} appears without proof in \cite[p. 18]{Ar15}.

\begin{lemma}\label{lem:arch}
The Archimedes spiral given in polar coordinates by
$r = c \theta$ for a fixed $c > 0$ has a constant distance
between its windings of $2 \pi c$.
\end{lemma}

\begin{proof}[Proof of Lemma~\ref{lem:arch}]
Two points on successive windings of the spiral are given by
$p_1 = (c \theta_1,\theta_1)$ and $p_2 = (c (\theta_1+2\pi),\theta_1+2\pi)$
for some $\theta_1 \geq 0$.
Then, using the law of cosines to compute distances in polar coordinates yields
\begin{align*}
  &d(p_1,p_2)^2 = (c\theta_1)^2 + (c(\theta_1+2\pi))^2 \\
  & \mkern+75mu
  - 2(c\theta_1)(c(\theta_1+2\pi))\cos(\theta_1+2\pi - \theta_1) \\
  &= c^2\theta_1^2 + c^2(\theta_1^2 + 4\pi \theta_1 + 4\pi^2)
  - 2c^2(\theta_1^2 + 2\pi \theta_1)
  = 4 \pi c^2
\end{align*}
so, taking the square root implies $d(p_1,p_2) = 2 \pi c$.
\end{proof}

\begin{corollary}\label{cor:dist}
For any $z \in \d$, there exists a point $p$ on the Archimedes spiral
$r = c \theta$ for a fixed $c > 0$ such that $|p| \leq |z|$ and
$d(z,p) < 2 \pi c$.
\end{corollary}

\begin{proof}[Proof of Corollary~\ref{cor:dist}]
Fix $z = (r,\theta) \in \d$ with $\theta \in [0,2\pi)$.
If $|z| = 0$ or $z$ is in the spiral, then the distance to the spiral is zero.
Therefore, assume that $z$ is not in the spiral.
Then there exists $\theta_1 \geq 0$ and an integer $n$ such that
$r \in (c \theta_1, c (\theta_1+2\pi))$ and $\theta_1 = \theta + 2 \pi n$.
Let $p_1 = (c \theta_1,\theta_1)$ and note that $p_1$ is in the spiral
and $|p_1| = c \theta_1 < r = |z|$.
By the above, we have
\begin{align*}
 &d(z,p_1)^2 = r^2 + (c\theta_1)^2 - 2r(c\theta_1)\cos(\theta + 2 \pi n - \theta)
  \\
  &= (r-c\theta_1)^2
  < (c(\theta_1+2\pi) - c\theta_1)^2
  = (2\pi c)^2.
\end{align*}
So taking the square root implies that $d(z,p_1) < 2 \pi c$.
As $|p_1| < |z|$ and $p_1$ is in the spiral, the claim follows.
\end{proof}

Corollary~\ref{cor:poles} and Corollary~\ref{cor:exists} above establish basic
geometric properties
about the particular pole selection along the spiral shown in
Fig.~\ref{fig:spiral}.
Lemma~\ref{lem:mon} below provides a technical result that is needed for the
proof of Corollary~\ref{cor:poles}.
For any positive even integer $m$,
let $c_m = \f{1}{2\sqrt{\pi\f{m+2}{2}}}$.  Then every pole
in the selection \eqref{eq:poles} lies along the Archimedes spiral
$r = c_m \theta$, as shown in Fig.~\ref{fig:spiral}.

\begin{lemma}\label{lem:mon}
For $x \in [1,\infty)$ and $m$ a positive even integer, let
\begin{align*}
  &g(x) = \\
  &\f{1}{\f{m}{2}+1} \left[2x+1 - 2\sqrt{x(x+1)}
    \cos(2\sqrt{\pi}(\sqrt{x+1}-\sqrt{x}))\right].
\end{align*}
Then $\lim_{x \to \infty} g(x) = (2 \pi c_m)^2$ and $g$ is monotonically
increasing over $x \in [1,\infty)$, so $g(x) < (2 \pi c_m)^2$ for all
  $x \in [1,\infty)$.
\end{lemma}

\begin{proof}[Proof of Lemma~\ref{lem:mon}]
By \cite[Appendix]{Ar15}
\begin{align*}
  \lim_{x \to \infty} g(x) = \f{\pi}{\f{m}{2}+1} = (2 \pi c_m)^2.
\end{align*}
To prove that $g$ is monotonically increasing over $[1,\infty)$,
it suffices to show that $g' = \de{g}{x} > 0$ over $[1,\infty)$.
Taking a derivative of $g$ leads to
\begin{align*}
  \left(\f{m}{2}+1\right)g'(x) &= 2 - \f{2x+1}{\sqrt{x(x+1)}}\cos(a_x)
  - a_x \sin(a_x) \\
  a_x &= 2\sqrt{\pi}(\sqrt{x+1}-\sqrt{x})
  = \f{2\sqrt{\pi}}{\sqrt{x+1}+\sqrt{x}}.
\end{align*}
For $z \in (0,1]$, define the function
\begin{align*}
  &h(z) = \left(\f{m}{2}+1\right) g'\left(\f{1}{z}\right) \\
  &= 2 - \f{2+z}{\sqrt{1+z}}\cos(a_z) - a_z\sin(a_z), \;
  a_z = 2\sqrt{\pi} \f{\sqrt{z}}{\sqrt{1+z}+1}.
\end{align*}
Taking a derivative leads to
\begin{align*}
  &h'(z)
    = \sqrt{\pi} \f{2+z-\sqrt{1+z}}{\sqrt{z}(1+z)(\sqrt{1+z}+1)}\sin(a_z) \\
  & -\left(\f{z}{2\sqrt{(1+z)^3}} + 2\pi \f{1}{\sqrt{1+z}(\sqrt{1+z}+1)^2}
  \right)\cos(a_z)
\end{align*}
Taking another derivative leads to $h''$ as shown in Fig.~\ref{fig:hdd}.
\begin{figure}
  \centering
  \includegraphics[width=0.35\textwidth]{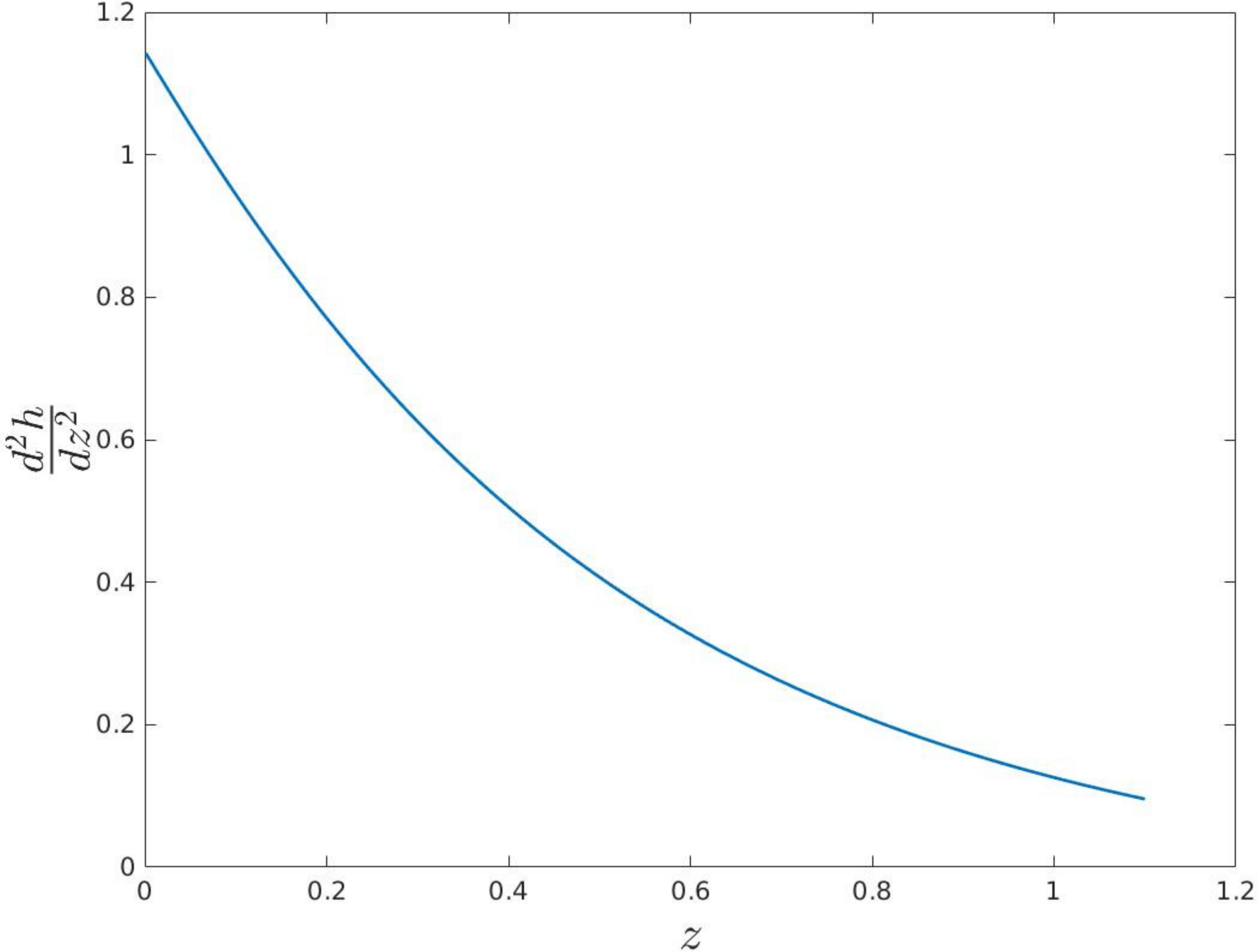}
  \caption{Graph of $h'' = \dd{h}{z}$.}
  \label{fig:hdd}
\end{figure}
Note that for $z \in [0,1]$, Fig.~\ref{fig:hdd} shows that $h''(z) > 0$.
This implies that $h'$ is monotonically increasing over $(0,1]$.
Using that $\lim_{z \to 0} \cos(a_z) = 1$,
$\lim_{z \to 0} \sin(a_z) = 0$, and by L'Hospital's rule we can
evaluate the limit $\lim_{z \to 0} h'(z) = 0$.
Thus, since $h'$ is monotonically increasing over $(0,1]$ and satisfies
$\lim_{z \to 0} h'(z) = 0$, this implies that $h'(z) > 0$ for all
$z \in (0,1]$.
In turn, this implies that $h(z)$ is monotonically increasing over
$z \in (0,1]$.
It is straightforward to see that $\lim_{z \to 0} h(z) = 0$
which, since $h(z)$ is monotonically increasing over $z \in (0,1]$,
implies that $h(z) > 0$ for $z \in (0,1]$.
Consequently, and since $x \in [1,\infty)$ implies that
$\f{1}{x} \in (0,1]$,
this implies that for every $x \in [1,\infty)$,
  $g'(x) = \f{1}{\f{m}{2}+1} h\left(\f{1}{x}\right) > 0$.
Thus, as $g'(x) > 0$ for all $x \in [1,\infty)$, $g$ is monotonically
increasing over $[1,\infty)$.
As $\lim_{x \to \infty} g(x) = (2 \pi c_m)^2$, this implies that
for $x \in [1,\infty)$ we have $g(x) < (2 \pi c_m)^2$.
\end{proof}

\begin{corollary}\label{cor:poles}
For any even integer $m > 0$,
selecting poles along the Archimedes spiral $r = c_m \theta$
according to \eqref{eq:poles}
for
$k \in \left\{1,...,\f{m}{2}\right\}$ 
implies that, for $k \in \left\{1,...,\f{m}{2}-1\right\}$
\begin{align*}
  \max \{d(p_k,p_{k+1}), d(p_{-k},p_{-(k+1)}), d(p_1,p_{-1})\} &< 2 \pi c_m \\
  \lim_{m \to \infty} \lim_{k \to m-1} d(p_k,p_{k+1}) &= 2 \pi c_m.
\end{align*}
\end{corollary}

Corollary~\ref{cor:poles} shows that in the limit, the distance between
successive poles in the Archimedes spiral approaches $2 \pi c_m$,
which is the same the distance between successive windings of the
spiral.  So, in this sense the distance between successive poles
is as close as possible to the distance between successive windings, which
was the goal for the pole selection along the spiral in an attempt to
approximate a uniform selection of poles over the unit disk.

\begin{proof}[Proof of Corollary~\ref{cor:poles}]
That $\lim_{m \to \infty} \lim_{k \to m-1} d(p_k,p_{k+1}) = 2 \pi c_m$ is proven
in \cite[Appendix]{Ar15}.
Since $2\theta_1 - 2\pi = 4\sqrt{\pi} - 2 \pi \in \left(0,\f{\pi}{2}\right)$,
$\cos(2\theta_1-2\pi) > 0$ so
\begin{align*}
  &d(p_1,p_{-1})^2 = r_1^2(2 - 2\cos(2\theta_1))
  = \f{(2 - 2\cos(2\theta_1-2\pi))}{\f{m}{2}+1} \\
  &< \f{2}{\f{m}{2}+1}
  < \f{\pi}{\f{m}{2}+1}
  = (2 \pi c_m)^2.
\end{align*}
So taking the square root implies that $d(p_1,p_{-1}) < 2 \pi c_m$.
For any $k \in \left\{1,...,\f{m}{2}\right\}$, $p_{-k} = \overline{p}_k$,
which implies by symmetry that
$d(p_{-k},p_{-(k+1)}) = d(p_k,p_{k+1})$
for all $k \in \left\{1,...,\f{m}{2}-1\right\}$.
Thus, to prove Corollary~\ref{cor:poles} it suffices to show that
for all $k \in \left\{1,...,\f{m}{2}-1\right\}$, $d(p_k,p_{k+1}) < 2 \pi c_m$.
Define $g(x)$ as in Lemma~\ref{lem:mon}, and note that for all
$k \in \left\{1,...,\f{m}{2}-1\right\}$, $g(k) = d(p_k,p_{k+1})^2$.
Since $k \in \left\{1,...,\f{m}{2}-1\right\} \subset [1,\infty)$, by
Lemma~\ref{lem:mon},
$d(p_k,p_{k+1})^2 = g(k) < (2\pi c_m)^2$
so taking the square root implies that $d(p_k,p_{k+1}) < 2 \pi c_m$.
\end{proof}

\begin{corollary}\label{cor:exists}
For any even integer $m > 0$, select poles along the Archimedes spiral
$r = c_m \theta$ according to Corollary~\ref{cor:poles}.
Then for any $z \in \d$, there exists $k \in \left\{1,...,\f{m}{2}\right\}$
such that $d(z,p_k) < 4 \pi c_m$ and either $|p_k| \leq |z|$ or $k = 1$.
\end{corollary}

\begin{proof}[Proof of Corollary~\ref{cor:exists}]
Fix $z \in d$.
By Corollary~\ref{cor:dist}, there exists a point $p$ in the spiral
$r = c_m \theta$
such that $|p| \leq |z|$ and $d(z,p) < 2 \pi c_m$.
Since $p$ is in the spiral, there exists $\hat{\theta} \geq 0$ such that
$p = (c_m \hat{\theta},\hat{\theta})$.
If $\hat{\theta} = \theta_k$ for some $k \in \left\{1,...,\f{m}{2}\right\}$,
then $p = p_k$ and $d(z,p_k) < 2 \pi c_m < 4 \pi c_m$ as desired.
So, it suffices to assume that $\hat{\theta} \neq \theta_k$ for any
$k \in \left\{1,...,\f{m}{2}\right\}$.
We claim that there exists $\hat{k} \in \left\{1,...,\f{m}{2}\right\}$
such that $d(p,p_{\hat{k}}) < 2 \pi c_m$ and either $|p_{\hat{k}}| \leq |z|$
or $|p_k| = r_1 = \sqrt{\f{k}{\f{m}{2}+1}}$.
From this claim the corollary follows since
$d(z,p_{\hat{k}}) \leq d(z,p) + d(p,p_{\hat{k}}) < 4 \pi c_m$.
So, it suffices to prove the claim.
To do so, we consider two cases.

Case 1: suppose that $\hat{\theta} \in (\theta_{\hat{k}},\theta_{\hat{k}+1})$
for some $\hat{k} \in \left\{1,...,\f{m}{2}\right\}$.
Note that $|p_{\hat{k}}| = c_m\theta_{\hat{k}} < c_m \hat{\theta} = |p| \leq |z|$.
Therefore, for Case 1 it suffices to show that $d(p,p_{\hat{k}}) < 2 \pi c_m$.
Define the function
\begin{align*}
  g(\theta) &= d((c_m\theta,\theta),p_{\hat{k}})^2 \\
  &=
  (c_m\theta)^2 + (c_m\theta_{\hat{k}})^2 - 2 (c_m\theta)(c_m\theta_{\hat{k}})
  \cos(\theta-\theta_{\hat{k}}).
\end{align*}
We claim that $g$ is monotonically nondecreasing on the interval
$\theta \in \left[\theta_{\hat{k}},\theta_{\hat{k}+1}\right]$.
This will then imply, since
$\hat{\theta} \in \left(\theta_{\hat{k}},\theta_{\hat{k}+1}\right)$ and
the square root function is monotonically nondecreasing, that
\begin{align*}
  d(p,p_{\hat{k}}) &= \sqrt{g(\hat{\theta})}
  \leq \sqrt{g(\theta_{\hat{k}+1})}
  = d(p_{\hat{k}+1},p_{\hat{k}})
  < 2 \pi c_m
\end{align*}
which will prove Case 1.
So, it remains to show that $g$ is monotonically nondecreasing
over $\left[\theta_{\hat{k}},\theta_{\hat{k}+1}\right]$.
To do so, it suffices to show that $\de{g}{\theta} \geq 0$
over $\left[\theta_{\hat{k}},\theta_{\hat{k}+1}\right]$.
First note that the function $k \to \sqrt{k+1}-\sqrt{k}$ is monotonically
decreasing in $k$ (to see this, note that $\sqrt{x+1} - \sqrt{x}$ has
negative derivative for $x > 0$).
By this fact and the definition of $\theta_k$ in \eqref{eq:poles},
for any $\theta \in \left[\theta_{\hat{k}},\theta_{\hat{k}+1}\right]$ we have
\begin{align*}
  \theta - \theta_{\hat{k}} &\leq \theta_{\hat{k}+1} - \theta_{\hat{k}}
  = \f{\sqrt{\hat{k}+1} - \sqrt{\hat{k}}}{(2 \sqrt{\pi})^{-1}}
  \leq \f{\sqrt{2}-\sqrt{1}}{(2 \sqrt{\pi})^{-1}}
  < \f{\pi}{2}
\end{align*}
where the final inequality follows from evaluation.
Thus, $\theta \in \left[\theta_{\hat{k}},\theta_{\hat{k}+1}\right]$ implies
$\sin(\theta-\theta_{\hat{k}}) \geq 0$ and
$\cos(\theta-\theta_{\hat{k}}) \geq 0$.
So
\begin{align*}
  \f{1}{c_m^2}\de{g}{\theta} &= 2 \theta
  - 2 \theta_{\hat{k}} \cos(\theta-\theta_{\hat{k}})
  + 2 \theta \theta_{\hat{k}} \sin(\theta-\theta_{\hat{k}}) \\
  &\geq 2 \theta - 2 \theta_{\hat{k}}
  + 2 \theta \theta_{\hat{k}} \sin(\theta-\theta_{\hat{k}})
  \geq 0
\end{align*}
since $\theta > \theta_{\hat{k}}$ and the other term is positive.
Thus, $\de{g}{\theta} \geq 0$
over $\left[\theta_{\hat{k}},\theta_{\hat{k}+1}\right]$, which concludes the
proof of Case 1.

Case 2: suppose $\hat{\theta} < \theta_1$.
If $\theta_1 - \hat{\theta} \geq \f{\pi}{2}$ then 
\begin{align*}
  &d(p,p_1)^2 = c_m^2(\hat{\theta}^2 + \theta_1^2 - 2\hat{\theta}\theta_1
  \cos(\hat{\theta}-\theta_1)) \\
  &\leq \f{\hat{\theta}^2 + \theta_1^2 + 2\hat{\theta}\theta_1}{(c_m)^{-2}}
  = \f{(\hat{\theta}+\theta_1)^2}{(c_m)^{-2}}
  \leq \f{\left(2\theta_1-\f{\pi}{2}\right)^2}{(c_m)^{-2}}
  < \f{(2\pi)^2}{(c_m)^{-2}}
\end{align*}
since $2\theta_1 - \f{\pi}{2} = 4 \sqrt{\pi} - \f{\pi}{2} < 2 \pi$,
so taking the square root implies that $d(p,p_1) < 2 \pi c_m$.
Otherwise, if $\theta_1 - \hat{\theta} \leq \f{\pi}{2}$ then
$\cos(\theta_1 - \hat{\theta}) \geq 0$, so
\begin{align*}
  &d(p,p_1)^2 = c_m^2(\hat{\theta}^2 + \theta_1^2 - 2\hat{\theta}\theta_1
  \cos(\hat{\theta}-\theta_1))
  \leq c_m^2(\hat{\theta}^2 + \theta_1^2) \\
  &\leq c_m^2(2 \theta_1^2)
  = c_m^2(\sqrt{2} \theta_1)^2
  \leq c_m^2(2 \pi)^2
\end{align*}
since $\sqrt{2}\theta_1 = 2\sqrt{2\pi} < 2\pi$, so taking the square
root implies that $d(p,p_1) < 2 \pi c_m$.
Thus, either way we have $d(p,p_1) < 2 \pi c_m$, which proves Case 2.
\end{proof}

\begin{proof}[Proof of Theorem~\ref{thm:spiral}]
Let $S \in \f{1}{z}\h$ and let $\mathcal{Q}$ be the poles of $S$.
The proof considers the spiral pole selection $\p$ of Corollary~\ref{cor:poles}.
For any pole $q$ in $\mathcal{Q}$, by Corollary~\ref{cor:exists} the distance
between $q$ and the closest pole in the spiral $\p$ is bounded by $4 \pi c_m$.
Traversing the spiral along neighboring poles
(i.e., $p_k$ and $p_{k+1}$), Corollary~\ref{cor:poles}
implies that the distance from $q$ to each successive pole in the spiral cannot
increase by more than $2 \pi c_m$.  This implies that the $m_q$ closest poles
in the spiral to $q$ are all within a distance of $(m_q+1)(2\pi c_m)$ from $q$.
As $(m_q+1)(2\pi c_m)$ is bounded by a constant times $\f{1}{\sqrt{m+2}}$,
this provides a bound on $D(\p)$.
Furthermore, for use of Theorem~\ref{thm:spiral} in the proofs of
Corollary~\ref{cor:spiral} and \cite[Corollary~1]{Fi22b},
additional work is required at the beginning of the proof to
only choose poles from the spiral within a ball of fixed radius
$\hat{r} \in (0,1)$, and at the end of the proof to only choose poles from the
spiral that are at least a fixed distance $\delta > 0$ away from all poles
in $\sigma$ which they are not approximating (see Assumption A5).

We begin by finding the desired radius $\hat{r} \in (0,1)$ of the ball for
intersecting with the spiral.
Let the poles $p_k$ and $p_{-k}$ be selected along the Archimedes
spiral according to Corollary~\ref{cor:poles} for the even integer $m$, and
define
$\hat{\p}_m = \mcup_{k=1}^{\f{m}{2}} \{p_k,p_{-k}\}$.
Let $m_{\max} = \max_{q \in \mathcal{Q}} m_q$.
Let $r_q = \max_{q \in \mathcal{Q}} |q|$ and let
$r_p = \max_{p \in \hat{\p}_{m_{\max}}} |p|$.
Let $\hat{r} = \max\{r_q,r_p\}$ and note that this implies that
$\mathcal{Q}, \hat{\p}_{m_{\max}} \subset \overline{B}_{\hat{r}}$
and $p_1 \in \overline{B}_{\hat{r}}$ for any $m$.
Let $\p_m = \hat{\p}_m \cap \overline{B}_{\hat{r}}$.
By the definition of the pole selection in Corollary~\ref{cor:poles}, and since
$\hat{\p}_{m_{\max}} \subset \overline{B}_{\hat{r}}$,
for any $m \geq m_{\max}$
we have $\p_m = \mcup_{k=1}^{k'} \{p_k,p_{-k}\}$ for some
$k' \in \left[\f{m_{\max}}{2},\f{m}{2}\right]$.
Thus, $|\p_m| \geq m_{\max}$.

Next we consider a pole $q$ of $\mathcal{Q}$, and bound the distance of its
$m_q$ closest approximating poles from the spiral.
Let $q \in \mathcal{Q}$.
By Corollary~\ref{cor:exists},
there exists $\hat{k} \in \left\{1,...,\f{m}{2}\right\}$
such that $d(q,p_{\hat{k}}) < 4 \pi c_m$
and either $|p_{\hat{k}}| \leq |q|$ or $\hat{k} = 1$.
The latter implies that $p_{\hat{k}} \in \overline{B}_{\hat{r}}$,
so $p_{\hat{k}} \in \p_m$.
As $2k' = |\p_m| \geq m_{\max} \geq m_q$, there exists a subset of
consecutive integers $S_{(q,m)} \subset \{-k', ..., -1,1,...,k\}$ (where
we define $-1$ and $1$ to be consecutive for this purpose) such that
$\hat{k} \in S_{(q,m)}$ and $|S_{(q,m)}| = m_q$.
Let $\p_m(q)$ denote the $m_q$ closest poles in $\p_m$ to $q$,
and let $\hat{i} = \argmax_{i \in S_{(q,m)}} |p_i-q|$.
Then by Corollary~\ref{cor:poles}
\begin{align*}
  &\hat{d}(q) = \max\nolimits_{p \in \p_m(q)} |p-q|
  \leq \max\nolimits_{i \in S_{(q,m)}} |p_i-q|
  = d(q,p_{\hat{i}}) \\
  &\leq d(q,p_{\hat{k}}) + d(p_{\hat{k}},p_{\hat{i}})
  < 4\pi c_m + \sum\nolimits_{i = \hat{k}}^{\hat{i}-1} d(p_i,p_{i+1}) \\
  &\leq 4 \pi c_m + |\hat{i}-\hat{k}|(2 \pi c_m)
  \leq 4 \pi c_m + (m_q-1)(2 \pi c_m) \\
  &= \f{(m_q+1)}{(2 \pi c_m)^{-1}}
  \leq \f{(m_{\max}+1)}{(2 \pi c_m)^{-1}}
  = \f{(m_{\max}+1)\sqrt{2 \pi}}{\sqrt{m+2}}.
\end{align*}
So
\begin{align*}
  D(\p_m) = \max_{q \in \mathcal{Q}} \hat{d}(q) \leq
  (m_{\max}+1)\sqrt{2 \pi} \f{1}{\sqrt{m+2}}.
\end{align*}
Substituting $m = 2n-2$ as in the statement of Theorem~\ref{thm:spiral} gives
\begin{align*}
  D(\p_n) \leq (m_{\max}+1)\sqrt{\pi} \f{1}{\sqrt{n}},
\end{align*}
so $\{p_n\}_{n=2}^\infty$ has geometric convergence rate $\f{1}{n^{1/2}}$.
Thus, $\{p_n\}_{n=2}^\infty$ exhibits geometric convergence, so by
Lemma~\ref{lem:equiv} it is space-filling.

Now we derive the desired $\delta > 0$ for ensuring the spiral poles are
at least distance $\delta$ from all poles of the plant which they are not
approximating.
As $\mathcal{Q}$ and $\sigma$ are finite (the latter by Assumption A5),
$\eta = \min_{q \in \mathcal{Q}, \lambda \in \sigma, \lambda \neq \sigma}
d(\lambda,q) > 0$
and $d(\lambda,q) \geq \eta$ for all such $\lambda \neq q$.
As $\{p_n\}_{n=2}^\infty$ exhibits geometric convergence, there exists
$N$ such that $n \geq N$ implies that $D(\p_n) \leq \f{1}{2} \eta$.
This implies that for $n \geq N$, $q \in \mathcal{Q}$, and $\lambda \in \sigma$
with $\lambda \neq q$, $d(\p_n(q),\lambda) \geq \f{1}{2}\eta$.
For $n \in \{2,...,N-1\}$ let
$d_n = \min_{q \in \mathcal{Q}, \lambda \in \sigma, \lambda \neq \sigma} d(\lambda,q)$.
Traversing downwards from $n = N-1$, let $\hat{N}$ denote the final value
of $n \in \{2,...,N-1\}$ for which Assumption A5 is satisfied by $\p_n$.
Then, for each $n \in \{\hat{N},...,N-1\}$, $d_n > 0$.
Let $d_N = \f{1}{2}\eta$, and let $\delta = \min_{n=\hat{N}}^N d_n > 0$.
Then for each $n \geq \hat{N}$, $q \in \mathcal{Q}$, and $\lambda \in \sigma$
with $\lambda \neq q$, $d(\p_n(q),\lambda) \geq \delta$.
The choices of $r$ and $\delta$ in this proof make it possible to apply
Theorem~\ref{thm:approx} uniformly to $\p_n$ for all $n \geq \hat{N}$
in the proofs of Corollary~\ref{cor:spiral} and \cite[Corollary~1]{Fi22b}.
\end{proof}

\begin{proof}[Proof of Corollary~\ref{cor:spiral}]
  By Theorem~\ref{thm:spiral}, $\{p_n\}_{n=1}^\infty$ is a sequence of poles
  with geometric convergence rate $\f{1}{n^{1/2}}$.
  Thus, the result follows by Theorem~\ref{thm:conv}.
\end{proof}

\section{Conclusion}\label{sec:con}

In Part I, SPA was introduced as a new Galerkin-type method for
finite dimensional approximations of Hardy space that is an alternative to
Lorentz series approximations such as FIR.
For any transfer function in Hardy space, approximation error bounds for 
SPA were provided which bound the Hardy space norms between the SPA and the
desired transfer function proportionally to the geometric distance between
their poles.
These were then used to show that the space of SPAs converges to the entire
Hardy space for any space-filling sequence of poles, and to provide a
uniform convergence rate that depends purely on the geometry of the SPA pole
selection.
Unlike with Lorentz series approximations such as FIR, where the convergence
rate is often related to the settling time of the optimal transfer function,
the uniform convergence of SPA ensures that it works well even when the
optimal transfer function has a long settling time, such as in systems with
large separation of time scales.
Finally, these results were specialized to the particular case of an Archimedes
spiral pole selection, for which an explicit uniform convergence rate was
provided.

In Part II, these results will be combined with SLS to develop a new control
design method with reduced suboptimality, guaranteed feasibility, ability to
include prior knowledge, formulation that requires solving only a single SDP,
and that avoids deadbeat control.
Furthermore, the proofs of the suboptimality certificates rely heavily on the
approximation and convergence results for SPA from Part I.
An example of optimal control design in Part II will demonstrate superior
performance of SPA compared to FIR for that case.

Future work will involve applying SPA to IOP, Youla parameterization, and
SLS with output feedback,
considerations of different space-filling sequences with favorable geometric
properties, and extensions to continuous time approximation methods.

\bibliographystyle{ieeetr}
\bibliography{refs}

\vspace{-20pt}

\begin{IEEEbiography}
  [{\includegraphics[width=1in,height=1.25in,clip,keepaspectratio]
      {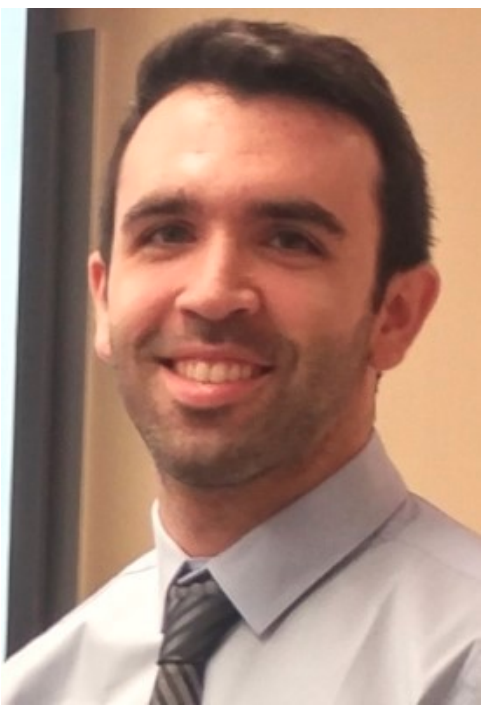}}]{Michael W. Fisher} is an Assistant Professor
  in the Department of Electrical and Computer Engineering at the University of
  Waterloo, Canada.  He was a postdoctoral researcher with
  the Automatic Control and Power System Laboratories at
  ETH Zurich.  He received his Ph.D. in Electrical Engineering:
  Systems at the University of Michigan, Ann Arbor in 2020, and a
  M.Sc. in Mathematics from the same institution in 2017. He received
  his B.A. in Mathematics and Physics from Swarthmore College in 2014.
  His research interests are in dynamics, control, and optimization of
  complex systems, with an emphasis on electric power systems.
  He was a finalist for the 2017 Conference on Decision and Control (CDC)
  Best Student Paper Award and a recipient
  of the 2019 CDC Outstanding Student Paper Award.
\end{IEEEbiography}

\vspace{-20pt}

\begin{IEEEbiography}
  [{\includegraphics[width=1in,clip,keepaspectratio]
      {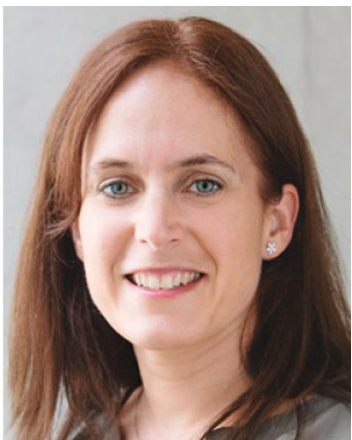}}]{Gabriela Hug} was born in Baden,
  Switzerland. She received the M.Sc. degree in electrical engineering
  and the Ph.D. degree from the Swiss Federal Institute of Technology
  (ETH), Zurich, Switzerland, in 2004 and 2008, respectively. After
  the Ph.D. degree, she worked with the Special Studies Group of Hydro
  One, Toronto, ON, Canada, and from 2009 to 2015, she was an
  Assistant Professor with Carnegie Mellon University, Pittsburgh, PA,
  USA. She is currently a Professor with the Power Systems Laboratory,
  ETH Zurich. Her research is dedicated to control and optimization of
  electric power systems.
\end{IEEEbiography}

\vspace{-20pt}

\begin{IEEEbiography}
  [{\includegraphics[width=1in,height=1.25in,clip,keepaspectratio]
      {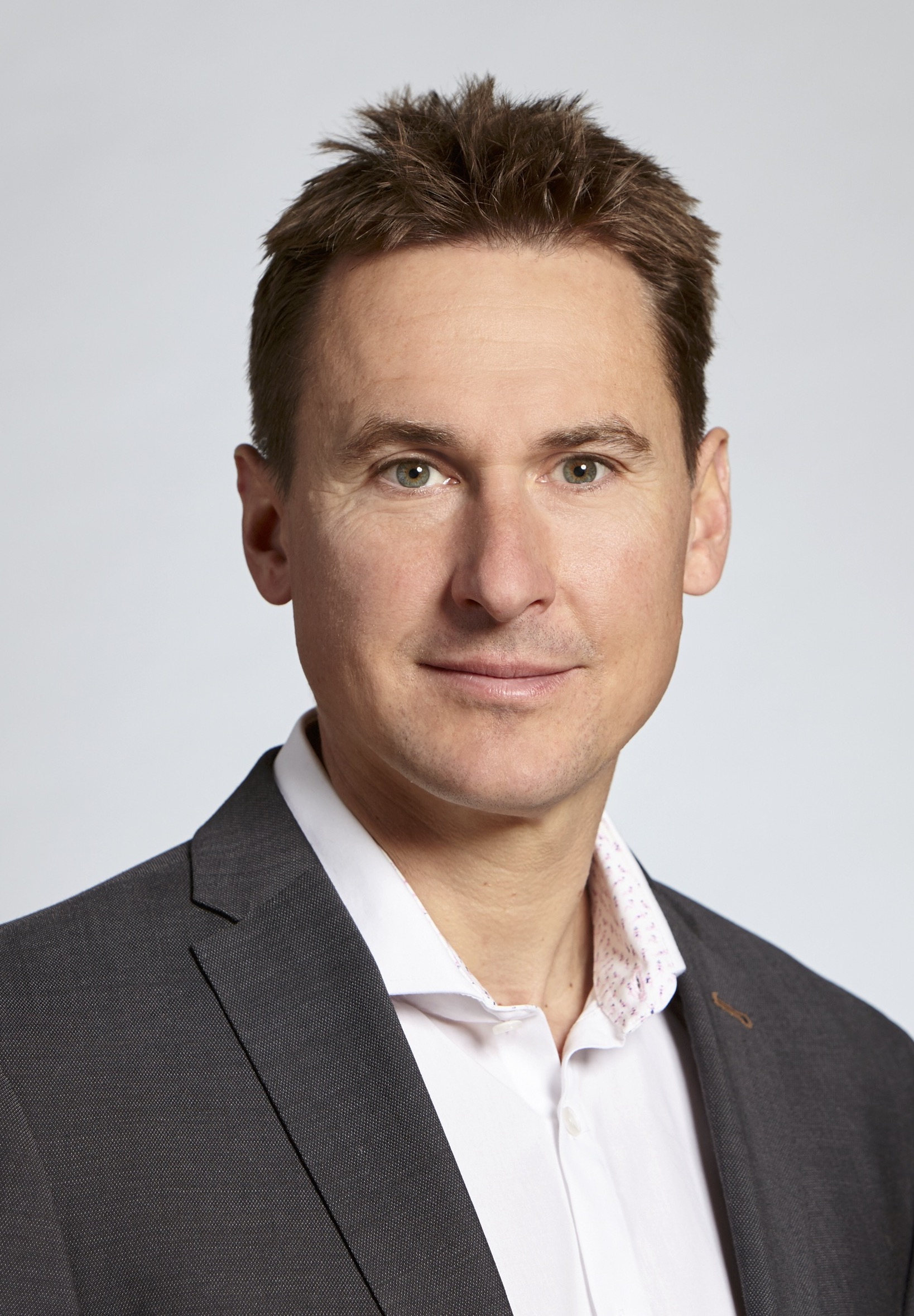}}]{Florian D\"{o}rfler} is an Associate
  Professor at the Automatic Control Laboratory at ETH Zurich,
  Switzerland, and the Associate Head of the Department of Information
  Technology and Electrical Engineering. He received his Ph.D. degree
  in Mechanical Engineering from the University of California at Santa
  Barbara in 2013, and a Diplom degree in Engineering Cybernetics from
  the University of Stuttgart, Germany, in 2008. From 2013 to 2014 he
  was an Assistant Professor at the University of California
  Los
  Angeles. His primary research interests are centered around control,
  optimization, and system theory with applications in network
  systems, especially electric power grids. He is a recipient of the
  distinguished young research awards by IFAC (Manfred Thoma Medal
  2020) and EUCA (European Control Award 2020).
  His students were
  winners or finalists for Best Student Paper awards at the European
  Control Conference (2013, 2019), the American Control Conference
  (2016), the Conference on Decision and Control (2020), the PES
  General Meeting (2020), the PES PowerTech Conference (2017), and the
  International Conference on Intelligent Transportation Systems
  (2021).
  He is furthermore a recipient of the 2010 ACC Student Best
  Paper Award, the 2011 O. Hugo Schuck Best Paper Award, the 2012-2014
  Automatica Best Paper Award, the 2016 IEEE Circuits and Systems
  Guillemin-Cauer Best Paper Award, and the 2015 UCSB ME Best PhD
  award.
\end{IEEEbiography}

\end{document}